\documentclass[showpacs,amsmath,amssymb]{revtex4}
 \usepackage{amsthm}
 \usepackage{graphicx} 
\newtheorem{thm}{Theorem}[section] 
\newtheorem{lemma}[thm]{Lemma} 
\newtheorem{cor}[thm]{Corollary} 
\newtheorem{prop}[thm]{Proposition} 
\newtheorem*{WCSG}{WCSG Bounds}
\newtheorem*{defn}{Definition}
\newtheorem*{conj}{Conjecture M}
\theoremstyle{remark}

\newtheorem*{remarks}{Remarks} 
\newcommand{\bu}{{\mbox{\boldmath $U$}}}
\newcommand{\C}{\mathbb{C}} 
\newcommand{\tr}{{\rm tr}}  
\newcommand{\bra}[1]{\langle#1|}
\newcommand{\ket}[1]{|#1\rangle} 
\newcommand{\braket}[2]{\langle#1|#2\rangle} 
\newcommand{\ketbra}[2]{|#1\rangle \langle#2|} 
\newcommand{\cvtwo}[2]{\left[\begin{array}{c}#1\\#2\end{array}\right]}%

\newcommand{\cvthree}[3]{\left[\begin{array}{c}#1\\#2\\#3\end{array}\right]}

\newcommand{\mthreethree}[9]{\left[\begin{array}{ccc}#1&#2&#3\\#4&#5&#6\\#7&#8&#9\end{array}\right]}

\newcommand{\vstrut}{{\rule{0in}{.14in}}}
\newcommand{\vvstrut}{{\rule{0in}{.18in}}}

\newcommand{\F}{\mathcal{F}}

\renewcommand{\v}{{\rm \bf v}}

\newcommand{\zero}{{\bf 0}}

\newcommand{\B}{E}
\newcommand{\CTT}{\C^{3\times 3}}
\newcommand{\omegadt}{{\omega_{d-3}}}
\newcommand{\omegad}{{\omega_{d-1}}}
\newcommand{\odd}{{\rm odd}}
\newcommand{\even}{{\rm even}}
\newcommand{\Pm}{{R}}
\begin{document}

\title{Deterministic Dense Coding and Entanglement Entropy}

\author{P.~S.~Bourdon}
\affiliation{Department of Mathematics,\\ Washington and Lee University, Lexington, VA 24450} 
\email{pbourdon@wlu.edu}
\author{E.~Gerjuoy}
\affiliation{Department of Physics and Astronomy,\\ University of Pittsburgh, Pittsburgh, PA 15260} 
\email{gerjuoy@pitt.edu}
\author{J.~P.~McDonald}
\affiliation{Departments of Mathematics and Physics,\\ Washington and Lee University, Lexington, VA 24450} 
\email{mcdonaldjp@wlu.edu}
\author{H.~T.~Williams}%
\affiliation{Department of Physics and Engineering,\\ Washington and Lee University, Lexington, VA 24450}%
 \email{williamsh@wlu.edu}

\date{\today}

\begin{abstract} We present an analytical study of the standard two-party deterministic dense-coding protocol, under which communication of perfectly distinguishable messages takes place via a qudit from a pair of non-maximally entangled qudits in pure state $|\psi\rangle$.
Our results include the following: (i) We prove that it is possible for a state $\ket{\psi}$ with lower entanglement entropy to  support the sending of a greater number of perfectly distinguishable messages than one with higher
entanglement entropy, confirming a result suggested via numerical analysis in Mozes {\it et al.}\ [Phys.\ Rev. A {\bf 71}, 012311 (2005)].  (ii) By explicit construction of families of local unitary operators,  we verify, for dimensions $d = 3$ and $d=4$, a conjecture of  Mozes {\it et al.}\  about the minimum entanglement entropy that supports the sending of $d + j$ messages, $2 \le j \le d-1$; moreover, we show that the $j=2$ and $j= d-1$ cases of the conjecture are valid in all dimensions.  (iii) Given that $\ket{\psi}$ allows the sending of $K$ messages and has $\sqrt{\lambda_0}$ as its largest Schmidt coefficient, we show that the inequality $\lambda_0 \le d/K$, established by  Wu {\it et al.}\  [ Phys.\ Rev.\ A {\bf 73}, 042311 (2006)],  must actually take the form $\lambda_0 < d/K$  if $K = d+1$, while our constructions of  local unitaries show that equality can be realized if $K = d+2$ or $K = 2d-1$.  \end{abstract}


 \maketitle

\section{Introduction}   Throughout this paper, we assume that Alice and Bob, located some distance apart,  each has initial control of one qudit from a two-qudit system in a pure state $\ket{\psi}$, with Schmidt representation 
\begin{equation}\label{GQDS}
\ket{\psi} = \sqrt{\lambda_0}\, \ket{0}_A\ket{0}_B + \sqrt{\lambda_1}\, \ket{1}_A\ket{1}_B + \cdots + \sqrt{\lambda_{d-1}}\, \ket{d-1}_A\ket{d-1}_B.
\end{equation}
 Here, $\sum_{k=0}^{d-1} \lambda_k = 1$, assuring normalization,  and we assume without loss of generality that 
 \begin{equation}\label{orderasmpt}
 \lambda_0 \ge \lambda_1\ge  \ldots \ge \lambda_{d-1}\ge
0.
 \end{equation}
   Let $H_A$ be the Hilbert space with orthonormal basis $\{\ket{0}_A, \ket{1}_A, \ldots, \ket{d-1}_A\}$, let $H_B$ be the Hilbert space with orthonormal
basis $\{\ket{0}_B, \ket{1}_B, \ldots, \ket{d-1}_B\}$, and let $H = H_A\otimes H_B$.   $H$ is the state space for the system that Alice and Bob share, with orthonormal basis $\ket{m}_A\ket{n}_B \equiv \ket{mn}$, $0\le m,n\le  d-1$, where the latter form for basis elements (letting order distinguish Alice and Bob's kets) will be used for notational convenience.   
 
 In the deterministic dense-coding protocol, originated by Bennett and Wiesner \cite{BW}, Alice prepares messages for Bob by applying to her qudit a physical operation actualizing a unitary operator (on $H_A$) chosen from a family $\{U_1, U_2,\ldots, U_K\}$ of {\it encoding unitaries}, and then she sends her qudit to Bob via a noiseless $d$-dimensional quantum channel. When Bob receives Alice's qudit, he measures the pair in an appropriate basis, obtaining a
message.  In order that the messages Alice sends  be perfectly distinguishable by Bob, the states $\{(U_j\otimes I)\ket{\psi}: j = 1, 2,\ldots K\}$ must be
(pairwise) orthogonal in the Hilbert space $H$.   Thus Alice's encoding unitaries are necessarily orthogonal in the sense that they generate orthogonal states in $H$, each corresponding to a message Alice may send to Bob.

If  $\ket{\psi}$ (given by (\ref{GQDS})) is maximally entangled, i.e., $\lambda_k = 1/d$ for $k = 0, \ldots, d-1$, then Alice may send $d^2$ perfectly
distinguishable messages to Bob (\cite{BW}; see, also \cite[\S II]{M}).  For non-maximally entangled states $\ket{\psi}$,  numerical work of Mozes {\it et al.}
\cite{M} suggests that for each $n$ in $\{d, d+1, \ldots, d^2 -2\}$ there is a collection of states that will support the sending of $n$ messages but not $n +1$.  As
in \cite{M}, we quantify entanglement level of a bipartite system in state $\ket{\psi}$ of  (\ref{GQDS})  using 
\begin{equation}\label{ES} S(\ket{\psi}) \equiv -\left(\sum_{j=0}^{d-l} \lambda_j\log_2(\lambda_j)\right),
\end{equation} 
which is the von Neumann entropy of either of the reduced density operators $\tr_B(\ketbra{\psi}{\psi})$ or $\tr_A(\ketbra{\psi}{\psi})$.    We call $S(\ket{\psi})$ the {\it entanglement entropy}  of $\ket{\psi}$.  

 Much of the work herein derives from numerical studies by Mozes {\it et al.} \cite{M}, designed to indicate the maximum number of  perfectly distinguishable messages that Alice can send to Bob as a function of the Schmidt coefficients of the initial state they share.   Before these studies were conducted, little was known about the message-bearing capacity of non-maximally entangled systems used for deterministic dense coding.  Based on their results, both numerical and analytic,  Mozes {\it et al.}\  \cite{M} formulate several conjectures, one of which is discussed in detail below. Also, their results suggest the following conjecture, not explicitly formulated in \cite{M}: for any two-qudit state $\ket{\psi}$ that supports a set of $K \ge d + 2$ encoding unitaries, $\lambda_0 \le d/K$.  The truth of this conjecture is a consequence of the following result due to Wu, Cohen, Sun, and Griffiths  \cite{W}: 
\begin{WCSG}Suppose that Alice and Bob share a two-qudit entangled state $\ket{\psi}$ with maximum Schmidt coefficient $\sqrt{\lambda_0}$. Let  $K\in
\{d, d+1, d + 2, \dots, d^2\}$.   If  $\ket{\psi}$ permits Alice to send $K$ messages to Bob via deterministic dense coding, then
\begin{equation}\label{WI}
\lambda_0 \le \frac{d}{K}.
\end{equation}
\end{WCSG}  
The work of \cite{M}  suggests that the upper bounds on $\lambda_0$ provided by inequality (\ref{WI}) cannot be decreased when $d+2 \le K \le d^2-2$, but that for $K = d+1$, the bound should instead be the smaller value $(d-1)/d$.  

Additional analytic work, supporting the validity of the numerical results in \cite{M},  is provided by Ji {\it et al.}\ \cite{J}, who  have proved that when $\ket{\psi}$   
is non-maximally entangled  the greatest number of messages that $\ket{\psi}$ can support is  $d^2-2$. Also,  \cite{M} itself contains  explicit constructions
of some families of encoding unitaries that are consistent with numerical data.  For example, the state
\begin{equation}\label{DPO}
\sqrt{(d-1)/d}\, \ket{00} + \sqrt{1/d}\, \ket{11} + 0\, \sum_{k=2}^{d-1} \ket{kk}
\end{equation} is shown to support the sending of $d+1$ messages and  
\begin{equation}\label{OH}
\sqrt{1/2}\, \ket{00} + \sqrt{1/2}\, \ket{11} +  0\, \sum_{k=2}^{d-1} \ket{kk}
\end{equation} is shown to support the sending of $2d$ messages.   Numerical work in \cite{M} suggests that the state (\ref{DPO}) is  the state with minimal
entanglement entropy (and largest $\lambda_0$ value) that supports the sending of $d+1$ messages and that the state (\ref{OH}) is the state with minimal entanglement entropy
(and largest $\lambda_0$ value) that supports the sending of $2d$ messages. More generally,  based on their numerical work, Mozes {\it et al} \cite[\S VII]{M}
conjecture  the following:

\begin{conj}\label{MC} The state (of a two qudit system) with minimal entanglement entropy supporting the sending of $d+j$ messages,  $j=2, 3, \ldots d$,  is
\begin{equation}\label{MES}
\ket{\psi_j} \equiv\sqrt{\frac{d}{d+j}}\, \ket{00} + \sqrt{\frac{j}{d+j}}\, \ket{11} + 0\, \sum_{k=2}^{d-1} \ket{kk}.
\end{equation}
\end{conj}
For a fixed value of  $\lambda_0\ge 1/2$, it is easy to check that the entanglement entropy (\ref{ES})  of a system in state (\ref{GQDS}) is minimized when
$\lambda_1 = 1-\lambda_0$ (and $0= \lambda_2= \lambda_3 = \cdots= \lambda_{d-1})$.  In addition, the entanglement entropy of $\sqrt{\lambda_0}\, \ket{00}
+\sqrt{1-\lambda_0}\, \ket{11}$ decreases as $\lambda_0$ increases ($\lambda_0\ge 1/2$).  Thus for a $j\in \{2, 3, \ldots, d\}$, the state $\ket{\psi_j}$ of
(\ref{MES}) represents the state with minimum entanglement entropy for which $\lambda_0 \le d/(d+j)$. 

Note that the WCSG Bounds are pertinent to Conjecture M, telling us that  if Alice wishes to send $K=d + j$ messages to Bob for some $j\in \{2, 3, \ldots, d\}$, then
$\ket{\psi_j}$ of (\ref{MES}) is the minimum-entanglement-entropy state that can possibly allow Alice to succeed.  The complete resolution of Conjecture M thus depends
on showing that  for  $j$ in $\{2, 3\, \ldots, d\}$, there are $d+j$ encoding unitaries for the state $\ket{\psi_j}$. 
  
 In general, when there exists a $\ket{\psi}$ with  $\lambda_0 = d/K$ that supports $K$ encoding unitaries, we will call the WCSG
bound $d/K$ {\it saturated} (a term used in \cite{W}).   Recall that Mozes {\it et al.} \cite{M} have constructed an orthogonal family of $2d$ encoding operators at
$\lambda_0 = 1/2$. This construction shows that the WCSG bound $\lambda_0=1/2$ is saturated and settles Conjecture M in the $j = d$ case.  Numerical work in \cite{M}
suggests that all the WCSG Bounds provided by inequality~(\ref{WI}) may be saturated except for  $\lambda_0 = d/(d^2-1)$ (shown not be saturated by Ji {\it et al.} \cite{J}) and $\lambda_0 = d/(d+1)$, which is the largest bound of interest in
that it provides information about where the transition from what is possible without entanglement (sending $d$ messages with a qudit) to  what is possible with
entanglement (sending more than $d$ messages with a qudit).  

   This paper's principal contributions to deterministic dense-coding theory include the following: 
\begin{itemize}
\item  \S II:  We present an alternate proof of the validity of  the WCSG Bounds, which provides useful
information about any family of $K$ encoding unitaries for $\ket{\psi}$ (of form (\ref{GQDS})), given that $\ket{\psi}$  supports $K$ fully distinguishable messages and $\lambda_0 = d/K$.  This  ``boundary  information'' is recorded as  Corollary~\ref{KC}. Using Corollary~\ref{KC}, we prove that the WCSG bound $d/(d+1)$ is not saturated.
\item \S III: We prove that there are two-qutrit systems with states $\ket{\psi_l}$ and $\ket{\psi_h}$ for which the entanglement entropy of the first is lower than that of the second, yet the first supports the sending of a greater number of perfectly distinguishable messages than the second.  This outcome, somewhat counterintuitive, was strongly suggested by numerical work of Mozes {\it et al.} \cite{M}.   The authors of \cite{M} attribute this possibility to the fact that the von Neumann entropy is an asymptotic quantity whereas the deterministic dense coding process can be carried out  with a single entangled pair.  Of course, it is the vector of Schmidt coefficients of  the bipartite-system state $ \ket{\psi}$ that ultimately determines the number of dense-coding messages $\ket{\psi}$ will support, and in general this vector is not determined by $S(\ket{\psi})$. 

 \item \S III \& \S IV:  We establish that Conjecture M is valid in dimensions $d=3$ and $d=4$ via explicit construction of  families of encoding unitaries. We also prove a non-extendibility result for certain families of encoding unitaries constructed using powers of the $d$-dimensional (unitary) shift operator $X_d$, defined by
 \begin{equation}\label{Xd}
 X_d\ket{j} = \ket{(j + 1)\hspace{-.1in}\mod d}, j = 0, 1, 2, \ldots, d-1.
 \end{equation}
  We show that if  $\{D_k\}_{k=0}^{d-1}$ is any collection of unitary diagonal operators, then the collection $\{X_d^kD_k\}_{k=0}^{d-1}$, which may serve as a collection of encoding unitaries for any two-qudit state $\ket{\psi}$,  cannot be extended to be a part of a larger family of encoding unitaries for $\ket{\psi}$ if $\lambda_0 > 1/2$. This result, as well as Corollary~\ref{KC}, helped us discover our families of encoding unitaries relevant to Conjecture M.
\item  \S IV:  We prove that the $j=2$ and $j = d-1$ cases of Conjecture M are valid for all dimensions (again, via explicit construction of encoding unitaries).
\end{itemize}  

\section{An Alternate Proof of the WCSG Bounds}

 We continue to work with
\begin{equation}\label{GQDSNS}
\ket{\psi} = \sum_{j=0}^{d-1} \sqrt{\lambda_j}\ket{jj},
\end{equation}
the  pure state of a two-qudit system shared by Alice and Bob, where $\ket{\psi}$'s Schmidt coefficients appear in decreasing order (\ref{orderasmpt}).    Recall that if $\{U_0, U_1, \ldots, U_{K-1}\}$ is a family of encoding unitaries acting on $H_A$, then  $\{(U_i\otimes I)\ket{\psi}: i  = 0, 1, 2, \ldots, K-1\}$ is an orthonormal set  in $H$, the members of which we call {\it orthogonal messages}.      

Throughout the paper, there are occasions when it will be useful to us to make additional assumptions about elements of sets of encoding unitaries.    The following well-known lemma provides an illustration.

  \begin{lemma}\label{IIF} If there is a family of $K$  encoding unitaries, then there is  also a family of $K$  encoding unitaries one of which
is the identity operator.
\end{lemma}
\begin{proof} Suppose that the two-qudit system used for dense coding is in state $\ket{\psi}$, given by (\ref{GQDSNS}).  Suppose that there is a family of $K$
orthogonal encoding unitaries $\{U_0, U_1, \ldots, U_{K-1}\}$ acting on $H_A$, so that
$
\ket{\psi_i}\equiv(U_i\otimes I)\ket{\psi}, i = 0, 1, \ldots, K-1,
$ constitute an orthonormal set in $H$.   For each $i$, let
$$
\ket{\phi_i} = (U_0^\dagger \otimes I)\ket{\psi_i}.
$$ Because unitary operators preserve inner products, we see that $\{\ket{\phi_i}: i = 0, 1, \ldots, K-1\}$ is also an orthonormal set in $H$.  In other words,
$\{U_0^\dagger U_0,  U_0^\dagger U_1, \ldots, U_0^\dagger U_{K-1}\}$ is also a family of $K$ orthogonal encoding unitaries.  Thus we have a family of $K$
orthogonal  encoding unitaries, one of which, $U_0^\dagger U_0$,  is the identity operator.
\end{proof}

Wu {\it et al.}\ \cite{W} use properties of partial-trace operators to establish the WCSG Bounds.  These bounds will be derived differently here, in a way that enables us to obtain information about any family of $K$ encoding unitaries for $\ket{\psi}$ if $\lambda_0 = d/K$.  

  Assume there are  $K\le d^2$ unitary operators $U_0, U_1, \ldots, U_{K-1}$ acting on $H_A$ such that $\ket{\psi_i} \equiv (U_i\otimes I)\ket{\psi}, i = 0, 1, \ldots,
K-1$,  form an orthonormal set in $H$.      We show $\lambda_0 \le d/K$.

 Let $S$ be the subspace of $H$ spanned by the orthonormal set $\{\ket{\psi_i}: i =0, \ldots, K-1\}$ so that 
 $$
 P_S \equiv  \sum_{i=0}^{K-1} \ket{\psi_i}\bra{\psi_i}
 $$
  is the projection operator for states in $H$ onto $S$. Let  $m\in \{ 0, 1, 2, \ldots, d-1\}$.  Consider
   \begin{eqnarray*} P_S \ket{m0} &=& \sum_{i=0}^{K-1}\ket{\psi_i}\braket{\psi_i}{m0}\\
    &=& \sum_{i=0}^{K-1} \ket{\psi_i}\bra{\psi}U_i^\dagger\otimes I\ket{m0}\\
      &=& \sum_{i=0}^{K-1}\sqrt{\lambda_0} \bra{0}U_i^\dagger\ket{m} \ket{\psi_i},
 \end{eqnarray*} 
from which it follows that the square of the norm of $P_S\ket{m0}$ is given by
\begin{equation}\label{FQC}
\|P_S\ket{m0}\|^2 =  \sum_{i=0}^{K-1}\lambda_0|\bra{m}U_i\ket{0}|^2 .
\end{equation} 
Sum both sides of the preceding equation from $m=0$ to $m = d-1$:  
\begin{eqnarray*}
\sum_{m=0}^{d-1} \|P_S\ket{m0}\|^2 &=& \lambda_0 \sum_{i=0}^{K-1}  \sum_{m=0}^{d-1} |\bra{m}U_i\ket{0}|^2 \\
    &= & \lambda_0 K,
  \end{eqnarray*}
  where to obtain the final equality we have used  $\sum_{m=0}^{d-1} |\bra{m}U_i\ket{0}|^2 = 1$, which follows from the unitarity of $U_i$. 
 Because $P_S$ is a projection operator,  $\|P_S\ket{m0}\| \le \|\ket{m0}\| = 1$; thus,
\begin{equation}\label{TGP}
 \lambda_0 K = \sum_{m=0}^{d-1} \|P_S\ket{m0}\|^2 \le d,
\end{equation}
so that $\lambda_0 \le d/K$, the desired result.

Note that by equation (\ref{TGP}) if there are $K$ encoding unitaries and $\lambda_0 = d/K$, then
$$
\sum_{m=0}^{d-1} \|P_S\ket{m0}\|^2 = d.
$$ Since $\|P_S\ket{m0}\| \le 1$ for each $m$, this can happen only if $\|P_S\ket{m0}\| = 1$ for each $m$.  Since $\ket{m0}$ is a unit vector, $\|P_S\ket{m0}\| =
1$  implies  $P_S\ket{m0} = \ket{m0}$; equivalently, $\ket{m0}$ belongs to $S$.  Thus the following is a corollary of our proof of the WCSG Bounds.
\begin{cor}\label{KC}  Let $K\in \{d, d+1, \ldots, d^2\}$, let $\lambda_0 = d/K$, and let 
$
\ket{\psi}  = \sum_{j=0}^{d-1} \sqrt{\lambda_j}\ket{jj},
$ where $\sqrt{\lambda_0}$ is the largest Schmidt coefficient of $\ket{\psi}$. Suppose that there are $K$ unitary operators $U_0, U_1, \ldots, U_{K-1}$ acting  on
$H_A$ such that $E\equiv \{(U_j\otimes I)\ket{\psi}: j = 0, 1, \ldots, K-1\}$ is an orthonormal set in $H$.  Then the linear span $S$ of the set   $E$ contains $\ket{m0}$ for each $m \in \{ 0, 1, \ldots, d-1\}$.
\end{cor}

As our first application of the preceding corollary, we prove the WCSG bound $d/(d+1)$ is not saturated.

\begin{prop}\label{BNS} If the state
$
\ket{\psi} = \sum_{j=0}^{d-1} \sqrt{\lambda_j}\ket{jj}
$
 supports the sending of $d + 1$ orthogonal messages via deterministic dense coding, then $\ket{\psi}$'s largest Schmidt coefficient, $\sqrt{\lambda_0}$, must satisfy
$$
\lambda_0 < \frac{d}{d+1}.
$$
\end{prop}
 \begin{proof} Assume there are $K = d+1$ encoding unitaries $\{U_0, U_1, \ldots, U_d\}$, so that
 $$
 \ket{\psi_i}\equiv (U_i\otimes I)\ket{\psi}, i = 0, 1, \ldots, d
 $$
 form an orthonormal set in $H$.  By Lemma~\ref{IIF}, we can and do assume that $U_0 = I$.   Suppose,  in order to obtain a contradiction, that
$$
\lambda_0 = \frac{d}{d+1}.
$$ 
Then Corollary~\ref{KC} tells us that the $d+1$ dimensional subspace $S$ spanned by $\{\ket{\psi_i}: i = 0, 1, \ldots, d\}$ contains the $d$ orthonormal
vectors $\ket{00}, \ket{10}, \ldots, \ket{(d-1)0}$.  The subspace $S$ also contains
$$ 
U_0\ket{\psi} = \ket{\psi} = \sum_{j=0}^{d-1} \sqrt{\lambda_j}\ket{jj}.
$$ 
Subtracting $\lambda_0\ket{00}$, which is in $S$, from $\ket{\psi}\in S$, we see 
$$
\ket{\gamma} =  \sum_{j=1}^{d-1} \sqrt{\lambda_j}\ket{jj}
$$ belongs to $S$ and that $\B\equiv\{\ket{00}, \ket{10}, \ldots, \ket{(d-1)0}, \ket{\gamma}\}$, being an  orthogonal set of $d + 1$ elements in $S$, must form an orthogonal
basis for $S$.  

Let $n$ be an arbitrary element in $\{1, 2, \ldots, d\}$.  Consider 
$$
\ket{\psi_n} = (U_n\otimes I)\ket{\psi} = \sum_{j=0}^{d-1}\sqrt{\lambda_j} (U_n\ket{j})\otimes \ket{j}.
$$ 
Since $\ket{\psi_n}$ is in $S$ and $\sqrt{\lambda_0}(U_n\ket{0})\otimes\ket{0}$  is also in $S$ (being a linear combination of  $\ket{00}, \ket{10},
\ldots, \ket{(d-1)0}$), we see, by subtraction, 
$$
\ket{\nu}\equiv \sum_{j=1}^{d-1}\sqrt{\lambda_j} (U_n\ket{j})\otimes \ket{j} 
$$ belongs to $S$.  Since $\B$ is an orthogonal basis for $S$ and $\ket{\nu}$ is orthogonal to all elements of $\B$ except $\ket{\gamma}$,  there is a scalar $c$ such that
$\ket{\nu} = c\ket{\gamma}$; that is,
$$
\sum_{j=1}^{d-1}\sqrt{\lambda_j} (U_n\ket{j})\otimes \ket{j} = c  \sum_{j=1}^{d-1} \sqrt{\lambda_j}\ket{j}\otimes\ket{j}.
 $$ 
 Since $\lambda_1$ is not zero  $(\lambda_0 < 1)$, the preceding equation shows that  the second column of the matrix  $\bu_n$, representing $U_n$ with respect to the basis
$\{\ket{0}, \ket{1}, \ldots, \ket{d-1}\}$, consists of $0$'s except for its second entry, which is $c$.  Hence $|c| = 1$.   Since the
second column of $\bu_n$ is orthogonal to its first column, we see the first column of $\bu_n$ must have a zero as its second entry.  In other words,
$\bra{1}U_n\ket{0} = 0$.  Since $n\in \{1, 2, \ldots d\}$ is arbitrary, and $U_0$ is the identity, we have
$$
\bra{1}U_i\ket{0} = 0 
$$ for $i = 0, 1, \ldots d$.  It follows that $\braket{10}{\psi_i} = 0$ for $i=0, 1, \ldots, d$;  and  thus $\ket{10}$ is orthogonal to $S$ and in $S$, a contradiction, completing the proof. 
\end{proof}

\section{Entanglement Entropy  and Number of Encoding Unitaries} 

In this section, we focus on deterministic dense coding based on a two-qutrit system.  In this context, we prove that a state with lower entanglement entropy
can support the sending of a greater number of perfectly distinguishable messages via deterministic dense coding than a state with higher entanglement entropy.  

 We consider a two-qutrit system, shared by Alice and Bob, in state
\begin{equation}\label{GQTS}
\ket{\psi} = \sqrt{\lambda_0}\ket{00} + \sqrt{\lambda_1}\ket{11} + \sqrt{\lambda_2}\ket{22},
\end{equation} where as usual, we assume 
\begin{equation}\label{sum1}
\lambda_0 + \lambda_1 + \lambda_2 = 1 \ \ {\rm and} \ \  \lambda_0 \ge \lambda_1\ge \lambda_2\ge 0.
\end{equation} Via primarily numerical methods, Mozes {\it et al.}\  \cite{M} assess how many orthogonal messages to Bob, Alice can generate, given their shared system is
in state $\ket{\psi}$.   Their results are nicely summarized in Figure~\ref{MozesD}, reproduced from \cite{M}.  

 \begin{figure}[h]  \centerline{  \includegraphics[height=2.5in]{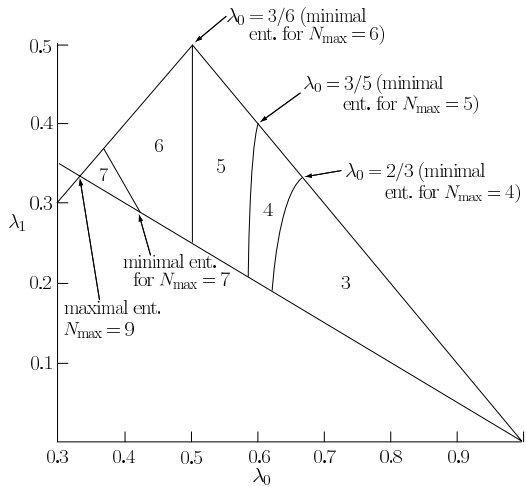}}
    \caption{(Mozes {\it el al.}) Numerical mapping of the maximum number ``$N_{\max}$''  of orthogonal encoding unitaries that the state $\ket{\psi}$ of (\ref{GQTS}) will
support, as a function of $\lambda_0$ and $\lambda_1$.   The equation  $\lambda_1 = (1-\lambda_0)/2$  defines the lower side of the triangular region bounding the numbered areas while region's upper side is piecewise defined as $\lambda_1 = \lambda_0$ for $1/3 \le \lambda_0 \le 1/2$ and $\lambda_1 = 1- \lambda_0$ for $1/2 <
\lambda_0\le 1$. \label{MozesD}} 
    \end{figure}

The figure suggests that if $\lambda_0$ equals,say, 0.51,  then $\ket{\psi}$ of (\ref{GQTS}) will support 5 orthogonal unitaries independent of the values of
$\lambda_1$ and $\lambda_2$ (satisfying (\ref{sum1})).  Note the role that the WCSG Bounds  play in the figure: it appears that $\lambda_0 = 3/7$ and $\lambda_0 = 3/5$ can be saturated
(as we explained in the Introduction, $\lambda_0 = 1/2$ is already known to be saturated)  and that $\lambda_0 = 3/4$ is not saturated (which is proved in
Proposition~\ref{BNS} above).   Note also that the boundary line between the 5 region and the 4 region appears to be curved slightly allowing its top point,
$\lambda_0 = 3/5, \lambda_1 = 2/5, \lambda_2 = 0$, to possibly be in the 5 region while the point  on the bottom boundary line directly below it,  $\lambda_0 =
3/5, \lambda_1= 1/5 = \lambda_2$, is not in the 5 region.  This is precisely what we prove below, and since the state corresponding to the top point,
$$
\ket{\psi_l}\equiv \sqrt{3/5}\, \ket{00} + \sqrt{2/5}\, \ket{11},
$$ has lower entanglement entropy than the bottom-line point
$$
\ket{\psi_h}\equiv \sqrt{3/5}\, \ket{00} + \sqrt{1/5}\, \ket{11} + \sqrt{1/5}\, \ket{22},
$$ we have a proof that a state with lower entanglement entropy can support more messages than one with higher entanglement entropy.
A plot of entanglement entropy over the triangular region of Figure~\ref{MozesD} appears  as Figure~\ref{EntropyPlot}.

\begin{figure}[h]  \centerline{  \includegraphics[height=2.7in]{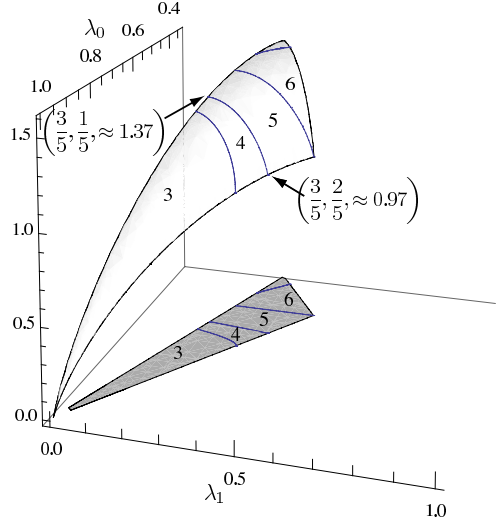}}
    \caption{A plot of von Neumann entropy $S(\ket{\psi})$ computed from Equation (\ref{ES}) for the state $\ket{\psi}=\sqrt{\lambda_0}\ket{00} + \sqrt{\lambda_1}\ket{11} + \sqrt{1-\lambda_0-\lambda_1}\ket{22}$ as $\lambda_0$ and $\lambda_1$ vary over the triangular region displayed in Figure~\ref{MozesD}} \label{EntropyPlot}
    \end{figure}

 A matrix point of view will facilitate our work.  As in the preceding sections, let $U$ be an encoding unitary operator  on $H_A$, a space which has orthonormal basis $\B\equiv\{\ket{0}, \ket{1}, \ket{2}\}$, where
we have dropped the subscripts of $A$ on the basis vectors.   {\em We will not  distinguish between $U$ and its representation as a $3\times 3$ matrix with
respect to $\B$}.  We will consider the $i,j$ entry of $U$ to be the entry in the $i$-th row and $j$-th column of $U$ so that 
$$ u_{ij} = \bra{i-1}U\ket{j-1}, 1\le i, j\le 3.
$$

We denote the collection of all $3\times 3$ matrices with complex entries $\CTT$.

Consider $\lambda_0$, $\lambda_1$, and $\lambda_2$ to be fixed nonnegative (real) numbers satisfying (\ref{sum1}).  Let 
\begin{equation}\label{Lform}
\Lambda= \mthreethree{\lambda_0}{0}{0}{0}{\lambda_1}{0}{0}{0}{\lambda_2}.
\end{equation}

\begin{defn}
We say that the unitary matrices $M$ and $U$ in $\CTT$ are $\Lambda$-{\it orthogonal} provided that
$$
\tr(\Lambda M^\dagger U) = 0, {\rm or,\ equivalently}, \  \tr(\Lambda U^\dagger M) = 0, 
$$ 
\end{defn}
 We say a set of unitary matrices is $\Lambda$-orthogonal provided its elements are pairwise $\Lambda$-orthogonal. The following proposition is easily verified (and appears as equation (7) in \cite{M}).

\begin{prop}\label{OP} The matrices $M$ and $U$ in $\CTT$ are $\Lambda$-orthogonal if and only if $(M\otimes I)\ket{\psi}$ and $(U\otimes I)(\ket{\psi}$ are
orthogonal vectors in $H$.   
\end{prop}

We consider unitary matrices in $\CTT$ to correspond to encoding operators that Alice may apply to her qutrit from the pair she shares with Bob.  Given their
qutrit pair is in state $\ket{\psi}$ of (\ref{GQTS}),    Proposition~\ref{OP} shows that the messages Alice encodes with unitary matrices $M$ and $U$  will be
perfectly distinguishable by Bob if and only if  $M$ and $U$ are $\Lambda$-orthogonal.  

Consider, for example, the situation where the two-qutrit system shared by Alice and Bob is fully entangled.  Here $\Lambda$ would be the diagonal matrix with
diagonal entries $1/3, 1/3, 1/3$ and it is easy to check that $\{I, X, X^2, Z, ZX, ZX^2, Z^2, Z^2X, Z^2X^2\}$ is a $\Lambda$-orthogonal family of 9 encoding
unitaries, where $I$ is the $3\times 3$ identity matrix, while $X$ and $Z$ are, respectively, shift and phase operators given by
$$ 
X = \mthreethree{0}{0}{1}{1}{0}{0}{0}{1}{0} \  ; \ \ Z = \mthreethree{1}{0}{0}{0}{\exp\left(\frac{2\pi i}{3}\right)}{0}{0}{0}{\exp\left(\frac{4\pi i}{3}\right)}.
$$ 
Also easy to check is that  $\{I, X, X^2\}$ is $\Lambda$-orthogonal independent of the values of $\lambda_0$, $\lambda_1$, and $\lambda_2$.  In fact, $\Lambda$-orthogonality for $\{I, X, X^2\}$ arises in the simplest possible way:  $X$, $X^2$ and $X^\dagger X^2$ (which equals $X$) each have main diagonal consisting of all zeros; when one multiplies such matrices by a diagonal matrix (on either the right or left), zeros remain on the main diagonal so that  the $\Lambda$-orthogonality of $\{I, X, X^2\}$ is clear.   As is pointed out in \cite{M},  the special properties of $\{I, X, X^2\}$ can be exploited to build $\Lambda$-orthogonal families of six unitaries if  $\lambda_0 \le 1/2$: specifically, for each $\Lambda$ of  the form (\ref{Lform}) with  $\lambda_0 \le 1/2$, there is a unitary diagonal matrix $D$ such that  $\{I, X, X^2, D, XD, X^2D\}$ will be $\Lambda$-orthogonal.   Our first result
in the qutrit context shows, however,  that if $\lambda_0 > 1/2$ then  the family $\{I, X, X^2\}$ cannot be expanded to a larger $\Lambda$-orthogonal family of  unitaries.   We prove something a bit more general: namely that if $B$ and $C$ are diagonal unitary matrices and $\lambda_0 > 1/2$, then $\{I, XB,  X^2C\}$ cannot be extended to a $\Lambda$-orthogonal family of 4 or more unitaries.   One can think of $XB$ and $X^2C$ as the two basic forms that a $3\times 3$ unitary matrix may assume if its main diagonal consists only of $0$'s.    The family
of five $\Lambda$-orthogonal unitaries that we construct below, for $\lambda_0 = 3/5, \lambda_1 = 2/5, \lambda_2 = 0$, contains no matrix having all zeros along its main diagonal.

\begin{prop} Suppose that  $\lambda_0> 1/2$ and that $B$ and $C$ are diagonal unitary $3\times 3$ matrices; then there is no unitary matrix $U$ such that the family $\{I, XB, X^2C, U\}$ is $\Lambda$-orthogonal.
\end{prop}
\begin{proof}   Let $\lambda_0 > 1/2$.  Let $B$ and $C$ be diagonal unitary matrices with $B$ having diagonal entries $\beta_1$, $\beta_2$, $\beta_3$ and $C$ having diagonal entries $\gamma_1$, $\gamma_2$,  $\gamma_3$.  Because $B$ and $C$ are unitary, each of the $\beta_j$'s and $\gamma_j$'s has modulus 1. 

  Suppose, in order to obtain a contradiction, that $U$ is a unitary matrix such that $\{I, XB,X^2C, U\}$ is
$\Lambda$-orthogonal.  Then, we have $\tr(\Lambda U) = 0$,  $\tr(\Lambda  (XB)^\dagger U) = 0$, $\tr(\Lambda (X^2C)^\dagger U) = 0$,  and these three equations expand, respectively, to
$$
\begin{array}{c}
\lambda_0 u_{11} + \lambda_1 u_{22} + \lambda_2 u_{33} = 0\ \\
\lambda_0 \beta_1^* u_{21} + \lambda_1\beta_2^* u_{32} + \lambda_2 \beta_3^* u_{13} = 0\ \\
\lambda_0 \gamma_1^*u_{31} + \lambda_1 \gamma_2^* u_{12} + \lambda_2 \gamma_3^* u_{23} = 0.
\end{array}
$$ The preceding system of equations may be written in the following vector form
\begin{equation}\label{VF}
\lambda_0 \v_1= - \lambda_1 \v_2 - \lambda_2\v_3,
\end{equation} where, 
$$
\v_1 = \cvthree{u_{11}}{\beta_1^* u_{21}}{\gamma_1^*u_{31}}, \v_2 = \cvthree{u_{22}}{\beta_2^*u_{32}}{\gamma_2^*u_{12}},\ {\rm and}\ \v_3 = \cvthree{u_{33}}{\beta_3^*u_{13}}{\gamma_3^*u_{23}}.
$$
 Since $U$ is unitary, each of its columns is a unit vector and it follows that $\v_1, \v_2$, and $\v_3$ are also unit vectors.  Taking the norm of both sides of
(\ref{VF}) and applying the triangle inequality, we obtain 
\begin{equation}\label{PC}
 \lambda_0 \le  \lambda_1 + \lambda_2,
 \end{equation}
 but we are assuming $\lambda_0 > 1/2$, which, since $\lambda_0 + \lambda_1 + \lambda_2 = 1$, makes $\lambda_1 + \lambda_2 < 1/2$. Thus (\ref{PC}) provides a
contradiction, completing the proof.
 \end{proof}

\begin{remarks}  (1) The preceding non-extendability result, as well as the argument that yields it, easily generalizes to dimensions $d> 3$.  Thus
if $X_d$ is the $d$-dimensional shift given by (\ref{Xd}) and if $\{D_k\}_{k=0}^{d-1}$ is any collection of unitary diagonal operators, then the family $\{X_d^kD_k\}_{k=0}^{d-1}$, which may serve as a collection of encoding unitaries for any two-qudit state $\ket{\psi}$,  cannot be extended to be a part of a larger family of encoding unitaries for $\ket{\psi}$ if $\lambda_0 > 1/2$. Other results in this section, such as Lemma~\ref{TIL}  and Proposition~\ref{FR} below, also have obvious higher-dimensional generalizations. (2) For the case $d=3$,  Cohen \cite{Ch} has established the following non-extendability result: Suppose that $\lambda_0 < 1/2$ and $D$ is a diagonal unitary matrix chosen so that $\{I, X, X^2, D, XD, X^2D,\}$ is $\Lambda$-orthogonal; then there is no unitary matrix $U$ such that  $\{I, X, X^2, D, XD, X^2D, U\}$ is also  $\Lambda$-orthogonal. 
\end{remarks}

\begin{lemma}\label{TIL} Suppose that $\lambda_0 >1/2$;  then no $\Lambda$-orthogonal family of  unitary matrices may contain the identity and a diagonal matrix
distinct from the identity.
\end{lemma}
\begin{proof} Suppose that $\lambda_0 >1/2$ and $I$ and $U$ are $\Lambda$-orthogonal for a diagonal unitary matrix $U$.  Then 
\begin{equation}\label{IUO}
\lambda_0 u_{11} + \lambda_1 u_{22} + \lambda_2 u_{33}  = 0.
\end{equation} Since $U$ is unitary and diagonal, each of $u_{11}, u_{22}$, and $u_{33}$ has modulus $1$.  Rewriting equation (\ref{IUO}) as $\lambda_0 u_{11} =
- \lambda_1 u_{22} - \lambda_2 u_{33}$ , taking absolute values of both sides, and applying the triangle inequality, we obtain $\lambda_0 \le \lambda_1 +
\lambda_2$, which cannot happen if $\lambda_0 > 1/2$.
\end{proof}

In our current context, Lemma~\ref{IIF} implies the following:
\begin{quotation}
{\em  If there exists a $\Lambda$-orthogonal family of $K$ unitary matrices in $\CTT$, then there also exists  a $\Lambda$-orthogonal family of $K$
unitary matrices in $\CTT$ one of which is the identity matrix.}
\end{quotation}
We depend upon a refinement of the preceding result that holds when $\lambda_1 = \lambda_2$. (Note the condition $\lambda_1 = \lambda_2$ defines the lower
line of the triangle in Figure~\ref{MozesD}.)

\begin{prop}\label{FR} Suppose that $\lambda_1 = \lambda_2$ and that there exists a $\Lambda$-orthogonal family of $K$ unitary matrices in $\CTT$.  Then there also
exists a $\Lambda$-orthogonal family of $K$ unitary matrices in $\CTT$ one of which is the identity and another of which has  $0$ as its $1,2$ entry.
\end{prop}
\begin{proof}  Under the hypotheses of this  proposition, we know that there is a family of $K$, $\Lambda$-orthogonal unitaries in $\CTT$ containing the identity
matrix.  Let this family be 
$$
\F = \{I, U, M_2, \ldots, M_{K-1}\}.
$$ As before, let $\Lambda$ be the diagonal matrix with diagonal entries $\lambda_0$, $\lambda_1$, and $\lambda_2$.  Because $\lambda_1 = \lambda_2$, it is easy
to check that $\Lambda$ commutes with any matrix of the form
$$ W = \mthreethree{1}{0}{0}{0}{w_{22}}{w_{23}}{0}{w_{32}}{w_{33}}.
$$ Set
$$ W = \mthreethree{1}{0}{0}{0}{\alpha}{-\beta^*}{0}{\beta}{\alpha^*}
$$ where 
$$
\ket{v}\equiv \cvtwo{\alpha}{-\beta^*} \ {\rm is\ a\ unit\ vector\ chosen\ to \ be\ orthogonal\ to}\  \ket{u} \equiv\cvtwo{u_{12}}{u_{13}}.
$$ Here, of course, $u_{jk}$  is the $j,k$ entry in the matrix $U$ from the family  $\F$.   Note that  $W$ is unitary and the $1,2$ entry of $WUW^\dagger$ is
$\braket{v}{u} = 0$.

We claim that the family obtained from $\F$ by left multiplication by $W$ and right multiplication by $W^\dagger$ satisfies the requirements of the proposition:
$$
\{WIW^\dagger, WUW^\dagger, WM_2W^\dagger, \ldots, WM_{k-1}W^\dagger\}
$$ consists of unitary matrices, the first element listed is the identity, the second has $0$ as its $1,2$ entry, and we claim that  its  elements are
$\Lambda$-orthogonal.   Consider, for example, 
\begin{eqnarray*}
\tr(\Lambda(WUW^\dagger)^\dagger (WM_2W^\dagger)) &=& \tr(\Lambda WU^\dagger M_2W^\dagger) \\ & =& \tr(W\Lambda U^\dagger M_2 W^\dagger)\\ &= & \tr(\Lambda
U^\dagger M_2) = 0,
\end{eqnarray*}   where we have used the fact that $W$ commutes with $\Lambda$, the cyclicity property of the trace, and the $\Lambda$-orthogonality of $\F$ to
obtain the final three equalities.
\end{proof}

Set 
$$
\Lambda_h = \mthreethree{3/5}{0}{0}{0}{1/5}{0}{0}{0}{1/5} \ {\rm and}\  \Lambda_l  = \mthreethree{3/5}{0}{0}{0}{2/5}{0}{0}{0}{0}.
$$ where we have used the subscripts $h$ and $l$ because, as we indicated earlier,  of the corresponding pair of states  $\ket{\psi_h}$ and $\ket{\psi_l}$, the
state $\ket{\psi_h}$ has the higher entanglement entropy and $\ket{\psi_l}$, the lower.  The following two propositions are the major results of this section, the first showing that $\ket{\psi_l}$ supports transmission of five orthogonal messages and the second showing that $\ket{\psi_h}$ supports at most four.

\begin{prop}\label{FiveFam} There are five $\Lambda_l$-orthogonal unitary matrices.
\end{prop}
\begin{proof} Let
$$
 A=\mthreethree{0}{1}{0}{1}{0}{0}{0}{0}{1} ,  U =\mthreethree{-\frac23}{0}{\frac{\sqrt{5}}{3}}{0}{1}{0}{-\frac{\sqrt{5}}{3}}{0}{-\frac{2}{3}},\ {\rm and}\  M =
\mthreethree{-\frac13}{-\frac{\sqrt{3}}{2}i}{\frac{\sqrt{5}}{6}}{\frac{1}{\sqrt{3}}i}{\frac12}{-\frac{\sqrt{5}}{2\sqrt{3}}i}{\frac{\sqrt{5}}{3}}{0}{\frac23}.
$$ Let $\F = \{I, A, M, M^*, U\}$, where $I$ is the $3\times 3$ identity matrix and $M^*$ is the matrix obtained from $M$ by taking the complex conjugate of each
of its entries.  The reader may verify that $\F$ is $\Lambda_l$-orthogonal (and consists of unitaries), completing the proof.
\end{proof}

The family $\F$   of five $\Lambda_l$-orthogonal matrices presented in the proof of the preceding Proposition was constructed with the assistance of
Corollary~\ref{KC}.  We began with the assumption that there is a family $\F\equiv\{I, A, U_1, U_2, U_3\}$ of five  $\Lambda_l$-orthogonal unitaries, where $I$ and $A$ are as in the proof of Proposition~\ref{FiveFam} and $U_1$, $U_2$, and $U_3$ were to be constructed (if possible).   Thus we knew that the linear span $S$ of the unitaries in $\F$ applied to $\ket{\psi_l}$ had to include the vectors $(I\otimes I)\ket{\psi_l} = \sqrt{3/5}\ket{00} +\sqrt{2/5} \ket{11}$,  $(A\otimes I)\ket{\psi_l} = \sqrt{3/5}\ket{10} + \sqrt{2/5}\ket{01}$, and, via Corollary~\ref{KC}, the vectors $\ket{00}, \ket{10},$ and  $\ket{20}$.  Because $S$ is closed under linear combinations, $\ket{11} =  \sqrt{5/2}[(I\otimes I)\ket{\psi_l}  - \sqrt{3/5}\ket{00}]$ and $\ket{01} = \sqrt{5/2}[(A\otimes I)\ket{\psi_l} - \sqrt{3/5}\ket{10}]$, were also necessarily in $S$.  We concluded that $\{\ket{00}, \ket{10}, \ket{20}, \ket{01}, \ket{11}\}$ is an orthonormal basis for the (five dimensional) subspace $S$. It follows that $\ket{21}$ must be orthogonal to $S$, which means in particular that for $j = 1, 2, 3$, 
$$
0 =\bra{21}(U_j \otimes I)\ket{\psi_l} = \sqrt{2/5} \cdot ({\rm the}\ 3, 2\ {\rm entry\ of}\ U_j),
$$
so that every encoding matrix in $\F$ necessarily had $0$ as its $3, 2$ entry.  With this knowledge and with some algebraic work, we arrived at the family $\F$ of Proposition~\ref{FiveFam}.

\begin{prop}  The number of $\Lambda_h$-orthogonal unitary matrices is less than or equal to four.
\end{prop}
\begin{proof} Because $\lambda_0 = 3/5$ for a system in state $\ket{\psi_h}$,   the inequality (\ref{WI}) for the WCSG Bounds shows that there are at most $5$
$\Lambda_h$-orthogonal unitary matrices.  Suppose, in order to obtain a contradiction, that there is a $\Lambda_h$-orthogonal family of 5 unitary matrices:  
$\F\equiv\{I, U, M, M_{2}, M_3\}$.  We assume that $U$ has $0$ as its $1,2$ entry, which we may do by Proposition~\ref{FR}.  

 Let $S$ be the linear span of $\{(I\otimes I)\ket{\psi_h}, (U\otimes I)\ket{\psi_h}, (M\otimes I)\ket{\psi_h},(M_2\otimes I)\ket{\psi_h},(M_3\otimes
I)\ket{\psi_h}\}$ and note that $S$ is a five-dimensional space.  By Corollary~\ref{KC}, $S$ must contain $\ket{00}, \ket{10}$, and $\ket{20}$, and hence $S$
also contains 
$$
\ket{\mu}\equiv \ket{\psi_h} - \sqrt{3/5}\, \ket{00} = \sqrt{1/5}\, \ket{11} + \sqrt{1/5}\, \ket{22} .
$$ 
Now consider $(U\otimes I)\ket{\psi_h}$, which too  belongs to  $S$; thus, 
\begin{eqnarray*}
\ket{\nu} &\equiv& (U\otimes I)\ket{\psi_h} - \sqrt{3/5}\, ( u_{11}\ket{00} + u_{21}\ket{10} + u_{31}\ket{20})\\
&  =& \sqrt{1/5}\, \left(\vstrut u_{12}\ket{01} + u_{22}\ket{11}+u_{32}\ket{21}+
u_{13}\ket{02}+u_{23}\ket{12} + u_{33}\ket{22}\right)
\end{eqnarray*}
 also belongs to $S$.  Recall  $u_{12} = 0$.  We claim that one of $u_{13}, u_{23}, u_{32}$ must be nonzero. Otherwise $u_{22}$ and $u_{33}$ are the only
nonzero entries in, respectively, the second and third columns of $U$ and, since $U$ is unitary, it follows that $u_{22}$ and $u_{33}$ are unimodular and that
$U$ is diagonal.  This contradicts Lemma~\ref{TIL}.  Thus at least one of $u_{13}, u_{23}, u_{32}$ is nonzero and it follows that $\B\equiv \{\sqrt{\frac{3}{5}}\, \ket{00}, \sqrt{\frac{3}{5}}\, \ket{10}, \sqrt{\frac{3}{5}}\, 
\ket{20}, \ket{\mu},\ket{\nu}\}$ is linearly independent.  Being a linearly independent  set of 5 elements in the 5 dimensional space $S$, we see that $\B$ is a
basis for $S$.  

Thus we may express $(M\otimes I)\ket{\psi_h}$ as a linear combination of the elements of $\B$:
$$ (M\otimes I)\ket{\psi_h} = m_{11} \sqrt{\frac{3}{5}}\, \ket{00} + m_{21}\sqrt{\frac{3}{5}}\, \ket{10} + m_{31}\sqrt{\frac{3}{5}}\, \ket{20} + q\ket{\mu} + c\ket{\nu}.
$$ It follows that $M$ must have the form
\begin{equation}\label{MForm} M = \mthreethree{m_{11}}{0}{cu_{13}}{m_{21}}{cu_{22} + q}{cu_{23}}{m_{31}}{cu_{32}}{cu_{33} + q}.
\end{equation}

We claim that both $c$ and $q$ are nonzero.  Because $M$ is unitary, its columns form an orthonormal basis of $\C^3$.  Thus $c$ and $q$ cannot both be zero,
because that would make both the second and third columns of $M$ the zero vector.  Suppose that $c = 0$.  Then $M$ would have to be a diagonal matrix, contradicting
Lemma~\ref{TIL}.  Suppose that $q = 0$.  Then the rightmost two columns of $M$ are simply $c$ times the corresponding columns of $U$.  Since both $U$ and $M$ are unitary, $|c| = 1$.   The first column of $M$ is determined to within a multiplicative constant by the other two, and thus can be written as a unimodular constant $\gamma$ times the first column of $U$.  Checking the $\Lambda$-orthogonality of $M$ and $U$ gives
$$
\tr(\Lambda M^\dagger U) = \lambda_0 \gamma^* + c^*(\lambda_1 + \lambda_2)
$$
Since $\lambda_0 > \lambda_1 + \lambda_2$, the preceding trace cannot be $0$, a contradiction.

  Therefore, we may assume that $M$ has the form (\ref{MForm}), where both $c$ and $q$ are nonzero. We can say more about the form of $M$.  Because $M$ and $U$ are both $\Lambda$-orthogonal to $I$, we have 
  $
  \frac{3}{5}m_{11} + \frac{1}{5}(c u_{22} + cu_{33} + 2q) = 0
  $  
  and
  $
  \frac{3}{5}u_{11} + \frac{1}{5}(u_{22} + u_{33}) = 0.
  $
    Combining these equations yields
  $$
  m_{11} = cu_{11} - \frac{2}{3}q.
  $$
  Because the second and third columns of $M$ are orthogonal, we have
$$ 0= |c|^2u_{22}^*u_{23} + q^*cu_{23} + |c|^2u_{32}^*u_{33} + c^*qu_{32}^* = |c|^2(u_{22}^*u_{23} + u_{32}^*u_{33}) + cq^*u_{23} + c^*q u_{32}^*.
$$ 
Because the second and third columns of $U$ are orthogonal, we conclude from the preceding equation that $0 =  cq^*u_{23} + c^*q u_{32}^*$, from which follows
that $|u_{23}| = |u_{32}|$, where we have used the fact that both c and q are nonzero.   Since the second column of $U$ is a unit vector and $u_{12} = 0$,
$|u_{22}|^2 + |u_{32}|^2 = 1$ so that $|u_{22}|^2 + |u_{23}|^2 = 1$.  It follows that $u_{21} = 0$.    Thus, because $u_{12}$ is also 0,  the inner product of
the first two columns of $U$ is  $u_{31}^*u_{32}$, which must equal 0.  There are two possibilities: either $u_{31} = 0$ or  $u_{32} = 0$.  We show each of these
possibilities leads to a contradiction, which completes the proof of the theorem.

Suppose that $u_{31} = 0$; then the only nonzero entry in the first column of $U$ is the first, which makes $|u_{11}| = 1$.  This is a contradiction because if
$u_{11}$ is unimodular, then $U$ cannot be $\Lambda_h$-orthogonal to  $I$.  If it were,  
$$ 0  = 3/5 u_{11} + 1/5 u_{22} + 1/5 u_{33},
$$  which implies $3/5 = |-1/5 u_{22} - 1/5 u_{33}| \le 2/5$, a contradiction.  Thus we must have  $u_{32} = 0$ and  $U$ takes the form
\begin{equation}\label{Uform} 
U =  \mthreethree{u_{11}}{0}{u_{13}}{0}{u_{22}}{0}{u_{31}}{0}{u_{33}}.
\end{equation}
 Note that $M_2$ and $M_3$  must have the same form as does $M$:
\begin{equation}\label{Mform}
\mthreethree{cu_{11} -(2/3)q}{0}{cu_{13}}{0}{cu_{22} + q}{0}{w}{0}{cu_{33} + q},
\end{equation}
where $c$, $q$, and $w$ are constants that depend on which of $M$, $M_2$, and $M_3$ is so represented.  Now, view the $3\times 3$ matrices $U, M, M_2$, and $M_3$ as vectors in the nine dimensional vector space $\CTT$.  Because $\{ U, M, M_2, M_3\}$ is a $\Lambda$-orthogonal set, the vector subspace $S$ of $\CTT$
spanned by 
$\{U, M, M_2, M_3\}$ is four dimensional. However, given the forms of these matrices from (\ref{Uform}) and (\ref{Mform}),  $S$ is also spanned by the 3 vectors in the following set
$$
\left\{ \mthreethree{-2/3}{0}{0}{0}{1}{0}{0}{0}{1}, \mthreethree{u_{11}}{0}{u_{13}}{0}{u_{22}}{0}{0}{0}{u_{33}}, \mthreethree{0}{0}{0}{0}{0}{0}{1}{0}{0}\right\},
$$ 
which is the contradiction that completes the proof.
\end{proof}
\section{Saturation of Some of the WCSG Bounds}

Proposition~\ref{FiveFam} of the preceding section shows that the state $\ket{\psi_l} = \sqrt{3/5}\, \ket{00} + \sqrt{2/5}\, \ket{11}$ is the two-qutrit state with
minimum entanglement entropy supporting the sending of 5 orthogonal messages through deterministic dense coding.  Thus the proposition resolves Conjecture M in
dimension $d = 3$  and shows that the WCSG  qutrit bound $\lambda_0 = 3/5$ is saturated.  
 
 \subsection{Dimension $d = 4$:  the  $d/(d+2)$ and $d/(2d-1)$ bounds}
   We now resolve Conjecture M for $d = 4$, which entails showing the WCSG Bounds of $\lambda_0 = 4/6$ and $\lambda_0 = 4/7$ are saturated.  Let 
\begin{equation}\label{MEQQS}
 \ket{\delta} = \sqrt{\lambda_0}\, \ket{00} + \sqrt{1-\lambda_0}\, \ket{11} + 0\, \ket{22} + 0\, \ket{33}.
\end{equation} We need to exhibit $6$ orthogonal encoding unitaries for  $\ket{\delta}$ with $\lambda_0 = 4/6 = 2/3$ and to exhibit $7$ orthogonal encoding
unitaries for $\ket{\delta}$ with $\lambda_0 = 4/7$.   We represent these encoding unitaries as $4\times 4$ matrices (with respect to the basis
$\{\ket{0},\ket{1},\ket{2}, \ket{3}\}$).  The encoding unitary matrices for $\ket{\delta}$, say $U_1$ and $U_2$, need to be $\Lambda$-orthogonal; i.e.,  $\tr(\Lambda U_1^\dagger U_2) = 0$,  where $\Lambda$ is the $4\times 4$ diagonal matrix whose diagonal entries, in order, are  $\lambda_0, 1-\lambda_0, 0, 0$.  Because only the initial two diagonal entries of $\Lambda$ are nonzero,  only the first two columns of unitary encoding matrices are relevant to $\Lambda$-orthogonality calculations. We record the full matrices, however, in
order for the reader to begin to see patterns that lead to higher-dimensional generalizations.  In all constructions below, Corollary~\ref{KC} was used to obtain
information about the form of members of encoding families of unitaries.

For $\lambda_0 = 4/6 = 2/3$,  our family of six $\Lambda$-orthogonal encoding unitaries for $\ket{\delta}$ is given by $\F_{4/6} = \{I, A, U_1(4), U_2(4), V_1(4), V_2(4) \}$, 
where $I$ is the identity and the remaining members of the family are
$$ 
A = \left[\begin{array}{cccc}0&1&0&0\\1&0&0&0\\ 0&0&1&0\\ 0&0&0&1\end{array}\right], U_1(4) = \left[\begin{array}{cccc}-\frac12&0& \frac{\sqrt{3}}{2}&0 \\
0&1&0&0\\ -\frac{\sqrt{3}}{2}&0&-\frac12&0\\ 0&0&0&1\end{array}\right], U_2(4) = \left[\begin{array}{cccc}-\frac12&0& \frac{\sqrt{3}}{2}&0\\ 0&1&0&0\\
\frac{\sqrt{3}}{2}&0&\frac12 &0\\ 0&0&0&1\end{array}\right],
$$
$$
 V_1(4) = \left[\begin{array}{cccc}0&1& 0&0\\ -\frac12&0&\frac{\sqrt{3}}{2}&0\\  0&0&0&1\\
-\frac{\sqrt{3}}{2}&0&-\frac12&0\\\end{array}\right], V_2(4) = \left[\begin{array}{cccc}0&1& 0&0\\ -\frac12&0&\frac{\sqrt{3}}{2}&0\\  0&0&0&1\\
\frac{\sqrt{3}}{2}&0&\frac12&0\\\end{array}\right].
$$ 

For $\lambda_0 = 4/7$,   the following is a $\Lambda$-orthogonal family of 7 unitary matrices:  $\F_{4/7}\equiv \{I, A_1, A_2, U, M_0, M_1,
M_2\}$, where $I$ is the identity and  the remaining members of the family are 
$$ 
A_1 = \left[\begin{array}{cccc} 0&0&1&0\\1&0&0&0\\0&1&0&0\\0&0&0&1\end{array}\right], A_2 =  \left[\begin{array}{cccc} 0&1&0&0\\
0&0&1&0\\1&0&0&0\\0&0&0&1\end{array}\right], U = \left[\begin{array}{cccc} -\frac{3}{4}&0&0&\frac{\sqrt{7}}{4}\\
0&1&0&0\\0&0&1&0\\-\frac{\sqrt{7}}{4}&0&0&-\frac{3}{4}\end{array}\right],
$$
$$
 M_j =\left[\begin{array}{cccc}
-\frac{1}{4}&-\frac{2}{3}\omega_3^{2j}&-\frac{2}{3}\omega_3^j &\frac{\sqrt{7}}{12}\\ \vvstrut \frac{1}{2}\omega_3^j&\frac{1}{3}&
-\frac{2}{3}\omega_3^{2j}&-\frac{\sqrt{7}}{6}\omega_3^j\\ \vvstrut \frac{1}{2}\omega_3^{2j}&-\frac{2}{3}\omega_3^j&
\frac{1}{3}&-\frac{\sqrt{7}}{6}\omega_3^{2j}\\ \vvstrut \frac{\sqrt{7}}{4}&0& 0&\frac{3}{4}\end{array}\right], 
\  {\rm where}\  \omega_3 = \exp(2\pi i/3)\ {\rm and}\ j =0, 1, 2.
$$  

\subsection{The $d/(2d-1)$ bound: dimensions $d \ge 5$}
  We now generalize the preceding constructions, i.e. those for the $d/(d+2)$ and $d/(2d-1)$ WCSG Bounds, to all dimensions $d \ge 5$.   We continue to assume that $\Lambda$ is a diagonal matrix whose 1,1 entry is
$\lambda_0$, whose 2,2 entry is $1-\lambda_0$, and all of whose other entries are 0's.   In this subsection, we construct $\Lambda$-orthogonal families of $2d - 1$ 
 unitaries for $\lambda_0 = d/(2d-1)$.   In the next, we construct  $\Lambda$-orthogonal  families of $d+2$ unitaries for $\lambda_0 = 2/(d+2)$.

 Let
$\lambda_0  = d/(2d -1)$.    For $d=4$,  the family  $\F_{4/7}$, exhibited above,  saturates the $d/(2d-1)$ bound.   Note that in $\F_{4/7}$, we can view the triple $I, A_1$, and $A_2$ as $A_0, A_1$ and $A_2$ where $A_j$ is the block diagonal matrix having the $j$-th power of the $3\times
3$ shift matrix,  $X^j$, in its upper left-hand corner and the $1\times 1$ matrix $[1]$ in its lower right-hand corner (with zeros elsewhere).   The $(d-1)$-dimensional shift operator $X_{d-1}$, defined by (\ref{Xd}),  will play a similar role in our  constructions for $d \ge  5$.   We are now in a position to show
that the WCSG bound $\lambda_0 = d/(2d-1)$ is saturated for every $d\ge 5$.   The construction for odd dimensions differs from that for even dimensions.

Key to the construction is the fact  that  if $n>1$ and $\omega$ is any $n$-th root of unity other than 1, then
\begin{equation}\label{Cmplx}
\sum_{j=0}^{n-1} \omega^j = 0.
\end{equation} Let $\omega_n = \exp\left(\frac{2\pi i}{n}\right)$.

 In what follows we will frequently describe only the first two columns of unitary encoding matrices.  As long as these two columns are orthonormal, there is a
unitary matrix that contains them.  Moreover, since only the first two entries of $\Lambda$ are nonzero,  only the first two columns of an encoding matrix are
relevant to $\Lambda$-orthogonality calculations.

 Suppose that $d\ge 5$ is odd.  We claim that  a family of $2d-1$ orthogonal encoding unitaries for the state 
 $\ket{\psi_{2d-1}}\equiv\sqrt{d/(2d-1)}\, \ket{00} + \sqrt{(d-1)/(2d-1)}\, \ket{11}$ is given by $\F_\odd = \{A_j\}_{j=0}^{d-2}\cup \{M_j\}_{j=0}^{d-2}\cup \{U\}$,
where $A_j$ is the  $d\times d$ block diagonal matrix with $X_{d-1}^j$ occupying its upper left-hand corner and the matrix $[1]$ occupying its lower right-hand corner (with zeros in
other locations) and where the first two columns of $U$ and $M_j$ are given by
 \begin{equation}\label{ADODD}
 U= \left[\begin{array}{cc}-\frac{d-1}{d}& 0\\0&1\\0&0\\ \vdots&\vdots\\ 0&0\\ -\frac{\sqrt{2d-1}}{d}&0\end{array}\right.,   M_j =
\left[\begin{array}{cc}-\frac{1}{d}& \frac{\sqrt{d}}{d-1}i\omegad^{(d-2)j}\\
-\frac{1}{\sqrt{d}}i\omegad^j&\frac{1}{d-1}\\-\frac{1}{\sqrt{d}}\omegad^{2j}&\frac{\sqrt{d}}{d-1}i\omegad^j\\ -\frac{1}{\sqrt{d}}i\omegad^{3j}&
\frac{\sqrt{d}}{d-1}\omegad^{2j}\\ \vdots&\vdots\\ -\frac{1}{\sqrt{d}}\omegad^{(d-3)j}& \frac{\sqrt{d}}{d-1}i\omegad^{(d-4)j}\\
-\frac{1}{\sqrt{d}}i\omegad^{(d-2)j}&\frac{\sqrt{d}}{d-1} \omegad^{(d-3)j}\\ \frac{\sqrt{2d-1}}{d}&0\end{array}\right. .
\end{equation}
 Note that all but the first and last entries in the first column of $M_j$ may be defined by the following formula:  
 $$
 \bra{k}M_j\ket{0} = -(-1)^{\left[\frac{k}{2}\right]}\frac{1}{\sqrt{d}}\, i^k\omegad^{kj}, \ \ k = 1, 2, \ldots, d-2,
 $$
 where $[k/2]$ represents the greatest integer less than or equal to $k/2$.  Note also that
 the second column of $M_j$ is determined by the first column as follows;   $\bra{0}M_j\ket{1} = - \frac{d}{d-1} \bra{(d-2)}M_j\ket{0}$ and for $k =  0, 1, 2,
\ldots, d-3$, 
$\bra{(k+1)}M_j\ket{1} =- \frac{d}{d-1}  \bra{ k}M_j\ket{0}$ (while $\bra{(d-1)}M_j\ket{1} = 0$).  Observe that these relationships between the entries in the first
and second columns of $M_j$ are precisely what make $M_j$, for any given $j\in \{0, 1, 2, \ldots, d-2\}$,  $\Lambda$-orthogonal to each $A_k$, $k\in \{0, 1, 2,
\ldots, d-2\}$, where $\lambda_0 = d/(2d-1)$.  It is easy to see that $\{A_k: k = 0, 1, 2, \ldots, d-2\}$ is $\Lambda$-orthogonal, in fact $A_k^\dagger A_\ell$
will have zeros as its $1,1$ and $2, 2$ entries as long as $k\ne \ell$.   The form of $U$ makes it easy to see that 
$$
\tr(\Lambda U^\dagger M_j) =   \frac{d}{2d-1}\left(\frac{d-1}{d^2} -\left(\frac{\sqrt{2d-1}}{d}\right)^2\right) + \frac{d-1}{2d-1}\frac{1}{d-1} = 0
$$
 independent of $j$. Also easy to see is that $\tr(\Lambda A_k^\dagger U) = 0$ independent of $k$ ($A_k^\dagger U$ will have $0$'s as its 1,1 and 2,2 entries for $k = 1, 2, \ldots, d-2$).   Clearly the first two columns of $U$ are unit vectors and are orthogonal so that $U$ extends to be a unitary $d\times d$ matrix. Also,
it's easy to see that the first two columns of $M_j$ are orthogonal unit vectors. (For the orthogonality calculation, it is  easier to take the
conjugate-transpose of the 2nd column times the first; terms will go to zero pairwise, the first by using the fact that the complex conjugate of 
$\omegad^{(d-2)j}$ is $\omegad^j$).    Thus, to complete our verification of the claim that $\F_\odd$ is a $\Lambda$-orthogonal family of unitaries, we must
check that for distinct $k$ and $\ell$, $M_k$ is $\Lambda$-orthogonal to $M_\ell$.  We have
\begin{equation*}
\begin{split}
\tr(\Lambda M_k^\dagger M_\ell) = \frac{d}{2d-1}\left(\frac{1}{d^2} + \frac{1}{d}\sum_{m=1}^{d-2} \left[\omegad^{\ell-k}\right]^m + \frac{2d-1}{d^2}\right) \\+
\frac{d-1}{2d-1}\left(\frac{1}{(d-1)^2} + \frac{d}{(d-1)^2}\sum_{m=1}^{d-2} \left[\omegad^{\ell-k}\right]^m\right).
\end{split}
\end{equation*}
The sums over $m$ on the right of the preceding equality are $-1$ by (\ref{Cmplx}) and the right side thus reduces to $1/(2d-1) - 1/(2d-1) = 0$,
as desired.  

  We continue to assume $\lambda_0 = d/(2d-1)$.  The construction of a $\Lambda$-orthogonal family of $2d-1$ encoding unitaries is more difficult when $d$ is
even:  the $4\times 4$ case doesn't entirely fit the general pattern and roots of unity of order d-3 as well as d-1 appear in the ``$M_j$''  construction
(because the pairwise cancellation that made the columns of $M_j$ orthogonal in the case of odd $d$ cannot work when $d$ is even).     Let $d\ge 6$ be even.   
We claim that  a family of $2d-1$, $\Lambda$-orthogonal encoding unitaries is given by $\F_\even = \{A_j\}_{j=0}^{d-2}\cup \{M_j\}_{j=0}^{d-2}\cup \{U\}$,
where, as before, $A_j$ is the  $d\times d$ block diagonal matrix with $X_{d-1}^j$ occupying its upper left-hand corner and $[1]$ occupying its lower right-hand corner,
and where the first two columns of $U$ are  given in  (\ref{ADODD}) and the first two columns of $M_j$ are given by
$$
 M_j = \left[\begin{array}{cc}-\frac{1}{d}& \frac{\sqrt{d}}{d-1}i\omegad^{(d-2)j}\\ -\frac{1}{\sqrt{d}}i\omegad^j&\frac{1}{d-1}\\ 
-\frac{1}{\sqrt{d}}i\omegad^{2j}&\frac{\sqrt{d}}{d-1}i\omegad^j\\ -\frac{1}{\sqrt{d}}i\omegadt\omegad^{3j}& \frac{\sqrt{d}}{d-1}i\omegad^{2j}\\ 
-\frac{1}{\sqrt{d}} i\omegad^{4j}& \frac{\sqrt{d}}{d-1}i\omegadt\omegad^{3j}\\-\frac{1}{\sqrt{d}}i\omegadt^2\omegad^{5j}& \frac{\sqrt{d}}{d-1}i\omegad^{4j}\\
\vdots&\vdots\\ -\frac{1}{\sqrt{d}}i\omegadt^{\frac{d-4}{2}}\omegad^{(d-3)j}& \frac{\sqrt{d}}{d-1}i\omegad^{(d-4)j}\\
-\frac{1}{\sqrt{d}}i\omegad^{(d-2)j}&\frac{\sqrt{d}}{d-1} i\omegadt^{\frac{d-4}{2}}\omegad^{(d-3)j}\\ \frac{\sqrt{2d-1}}{d}&0\end{array}\right. \ .
$$
 Note that all entries in the first column of $M_j$ except the first and the last are given by the following formula
 $$
 \bra{k}M_j\ket{0} = -\frac{1}{\sqrt{d}}\, i\, \omegadt^{\frac{(k-1)((-1)^{k+1} + 1)}{4}}\omegad^{kj}, k = 1, 2, \ldots, d-2.
 $$ 
 The second column of $M_j$ is determined by the first in the same way as before (i.e., in the $d$  is odd case), and, just as before, this means that $M_j$, for
any given $j\in \{0, 1, 2, \ldots, d-2\}$, is $\Lambda$-orthogonal to each $A_k$, $k \in \{0, 1, 2, \ldots, d-2\}$.  Just as before the $A_j$'s are (pairwise) 
$\Lambda$-orthogonal as are the $A_j$'s and $U$ as well as $U$ and  the $M_j$'s.   It is easy to check that the first two columns of $M_j$ are unit vectors.  To verify that they are orthogonal we
compute the conjugate transpose of the second column times the first:
\begin{equation}\label{EVENP}
\begin{split}
\left[\frac{1}{\sqrt{d}(d-1)} i (\omegad^*)^{(d-2)j} - \frac{1}{\sqrt{d}(d-1)} i \omegad^j\right] \rule{1in}{0in}  \\ - \frac{1}{d-1}\omegad^j 
\left(\sum_{m=0}^{\frac{d}{2}-2}\omegadt^m + \sum_{m=1}^{\frac{d}{2}-2}(\omegadt^*)^m\right).
 \end{split}
\end{equation}
 The expression in square brackets is zero since $(\omegad^*)^{(d-2)j} = \omegad^{j}$.   Also, using
 $$
  (\omegadt^*)^m = \exp\left(2\pi i\right)\exp\left(\frac{-2\pi i m}{d-3}\right) = \exp\left(2\pi i \frac{d-3 - m}{d-3}\right),
 $$ 
 we see that the second sum over $m$  in (\ref{EVENP}) equals 
$$
\sum_{m=\frac{d}{2} - 1}^{d-4} \omegadt^m
$$
and so both sums over $m$ in (\ref{EVENP}) combine to the sum of $\omegadt^m$ from $m=0$
 to $m=d-4$, which is 0 by (\ref{Cmplx}).  Thus the columns of $M_j$ are indeed orthogonal.   
 
 To complete the verification that $\F_\even$ is
$\Lambda$-orthogonal, we must check that for distinct $k$ and $\ell$, $M_k$ is $\Lambda$-orthogonal to $M_\ell$.  The calculation here is the same as that for
the case of odd $d$, the point being that in the product $M_k^\dagger M_\ell$ all $d-3$ roots of unity will be multiplied by their conjugates and  thus reduce
to $1$.   
 
 \subsection{The $d/(d+2)$ bound: dimensions $d\ge 5$}
We now turn to the  task of establishing that the WCSG bound of $\lambda_0=d/(d+2)$ is saturated for all $d\ge 5$. For $d=2$, the saturation of this bound is quite well known, e.g, the $2\times 2$ identity matrix, together with the Pauli matrices constitute a $\Lambda$-orthogonal family of $4$ encoding unitaries. Recall that we have already established saturation for the $d=3$ (Proposition~\ref{FiveFam}) and $d=4$ cases (via the family $\F_{4/6}$ exhibited earlier in this section).    Because our construction of orthogonal unitaries for the $\lambda_0 = d/(d+2)$ bound is inductive,  we modify our notation to reflect dependence on dimension $d$.       For the remainder of this section,
\begin{itemize}
\item $\Lambda_d$ is the $d\times d$ diagonal matrix whose diagonal entries, in order, are  $\lambda_0 = d/(d+2), \lambda_1 = 1-\lambda_0 = 2/(d+2), \lambda_2 = 0, \ldots, \lambda_{d-1} = 0$.
\end{itemize}

  Once again, our saturating families differ in form depending on whether $d$ is even or odd.    We will begin with the even case.  First, we introduce some important notation.

 In what follows, we use $[M]_j$ to denote the $j$-th column of the matrix $M$,  $I_n$ to denote the $n\times n$ identity matrix, $A_n$ to denote the matrix obtained by interchanging the first two columns of $I_n$, and $\zero_n$ to denote a column containing $n$ zeros.  A permutation matrix plays a role in our work: $\Pm$ is a (unitary) permutation matrix such that $\Pm M$  moves the second row of $M$ to the last row and shifts all rows initially beneath the second row up one.  Finally,  in all constructions below, we continue to exhibit only the first two columns of our matrices; if these columns are orthonormal, the given matrix may be extended to be unitary.   
  
  Having already defined $\F_{4/6}$, for each even $d\ge 6$ we define, inductively,  
  \begin{equation}\label{ddte}
  \F_{d/(d+2)}=\{I_d,A_d\}\cup\{U_j(d)\}_{j=1}^{d/2}\cup\{V_j(d)\}_{j=1}^{d/2},
 \end{equation}
 where
\begin{equation}\label{UFVF}
U_{d/2}(d)=\left[\begin{array}{cc}-\frac{2}{d}&0\\0&1\\  \sqrt{1-(\frac{2}{d})^2}\ \left[I_{\frac{d}{2}-1}\right]_1&\zero_{\frac{d}{2}-1}\\ \zero_{\frac{d}{2}-1}& \zero_{\frac{d}{2}-1} \end{array}\right.,
V_{d/2}(d)=\left[\begin{array}{cc}0&1\\-\frac{2}{d}&0\\ \zero_{\frac{d}{2} - 1}& \zero_{\frac{d}{2} - 1}\\ \sqrt{1-\left(\frac{2}{d}\right)^2}\ \left[I_{\frac{d}{2} - 1}\right]_{\frac{d}{2}-1} & \zero_{\frac{d}{2}-1}\end{array} \right.,
\end{equation}
  and where the remaining $d/2 - 1$ matrices in the $U$ and $V$ collections are constructed from the corresponding collections from the family $\F_{(d-2)/d}$ as follows: for $j = 1, 2, \ldots, d/2-1$,
\begin{equation}\label{UJVJ}
U_j(d)=\left[\begin{array}{cc}-\frac{2}{d}&0\\0&1\\\sqrt{1-(\frac{2}{d})^2}\Pm\left[U_j(d-2)\right]_1&\zero_{d-2} \end{array}\right.,
V_j(d)=\left[\begin{array}{cc}0&1\\-\frac{2}{d}&0\\\sqrt{1-(\frac{2}{d})^2}\Pm\left[V_j(d-2)\right]_1&\zero_{d-2}\end{array}\right..
\end{equation}
To illustrate how this inductive process works, we build  the family $\F_{6/8}$ using the family $\F_{4/6}$ presented earlier in this section.  The family $\F_{6/8}$ consists of $I_6$, $A_6$, 
$$
U_3(6)= \left[\begin{array}{cc} -\frac13& 0 \\ 0&1\\ \frac{\sqrt{8}}{3} & 0\\ 0&0\\ 0&0 \\0&0\end{array}\right., V_3(6)  = \left[\begin{array}{cc}0& 1 \\ -\frac13&0\\ 0& 0\\ 0&0\\  0 &0 \\ \frac{\sqrt{8}}{3}&0\end{array}\right.,  U_1(6) =  \left[\begin{array}{cc} -\frac13& 0 \\ 0&1\\ \frac{\sqrt{8}}{3}\left[\begin{array}{c} -\frac12\\-\frac{\sqrt{3}}{2}\\0\\0\end{array}\right] & \begin{array}{c}0\\0\\0\\0\end{array} \end{array}\right.
$$
$$
U_2(6) =  \left[\begin{array}{cc} -\frac13& 0\\ 0&1\\ \frac{\sqrt{8}}{3}\left[\begin{array}{c} -\frac12\\ \frac{\sqrt{3}}{2}\\0\\0\end{array}\right] & \begin{array}{c}0\\0\\0\\0\end{array} \end{array}\right., 
V_1(6) =  \left[\begin{array}{cc}0& 1 \\  -\frac13&0\\ \frac{\sqrt{8}}{3}\left[\begin{array}{c}0\\0\\-\frac{\sqrt{3}}{2}\\ -\frac12 \end{array}\right] & \begin{array}{c}0\\0\\0\\0\end{array} \end{array}\right., V_2(6) =  \left[\begin{array}{cc}0& 1 \\  -\frac13&0\\ \frac{\sqrt{8}}{3}\left[\begin{array}{c}0\\0\\ \frac{\sqrt{3}}{2} \\ -\frac12 \end{array}\right] & \begin{array}{c}0\\0\\0\\0\end{array} \end{array}\right. .$$
The reader may verify directly that the family $\F_{6/8}$ exhibited above is $\Lambda_6$-orthogonal and consists of unitary matrices.  We now turn to the proof that the family $\F_{d/(d+2)}$ defined  by (\ref{ddte}) is $\Lambda_d$ orthogonal for every even $d\ge 6$. 

Our formal, inductive argument begins with the $d=4$ case.   We have exhibited a family $\F_{4/6}$ of  six $\Lambda_4$-orthogonal encoding unitaries for $\lambda_0 = 4/6$ of the form $\{I_4,A_4\}\cup\{U_j(4)\}_{j=1}^{2}\cup\{V_j(4)\}_{j=1}^{2}$.  Suppose that for some even $d \ge 4$  the collection $\F_{d/(d+2)} = \{I_d,A_d\}\cup\{U_j(d)\}_{j=1}^{d/2}\cup\{V_j(d)\}_{j=1}^{d/2}$ is $\Lambda_d$-orthogonal and consists of unitaries.   We claim that the family $\F_{(d+2)/(d+4)} = \{I_{d+2},A_{d+2}\} \cup\{U_j(d+2)\}_{j=1}^{\frac{d}{2}+1}\cup\{V_j(d+2)\}_{j=1}^{\frac{d}{2} +1}$ is  a $\Lambda_{d+2}$-orthogonal family of unitaries.  Establishing this claim, completes the proof.  

Clearly the first two columns of $U_{\frac{d}{2}+1}(d+2)$ and $V_{\frac{d}{2} + 1}(d+2)$, defined by (\ref{UFVF}),  are orthonormal (and hence these columns may be augmented to create unitary matrices).   For $U_j(d+2)$ and $V_j(d+2)$, defined by (\ref{UJVJ}) orthogonality of the first two columns is obvious and normality follows from the fact that  $\Pm\left[U_j(d-2)\right]_1$ and $\Pm\left[V_j(d-2)\right]_1$ are unit vectors.  We must verify $\Lambda_{d+2}$-orthogonality of the members of $\F_{(d+2)/(d+4)}$.  

Clearly $I_{d+2}$ and $A_{d+2}$ are $\Lambda_{d+2}$-orthogonal.  
It is easy to see, due to the placement of zeros, that each $U_j(d+2)$ matrix will be $\Lambda_{d+2}$-orthogonal to each $V_k(d+2)$ matrix; $j,k\in \{1, 2, \ldots, d/2 +1\}$.  It is easy to check that each $U_j(d+2)$ as well as each $V_k(d+2)$  matrix is $\Lambda_{d+2}$-orthogonal to both $I_{d+2}$ and  $A_{d+2}$.    The remaining orthogonality checks require a little more effort.  

Let $n = d/2$ (so that, e.g., the 1,1 entry in the first column of each $U_j(d+2)$ matrix becomes $-1/(n+1)$). Let $\lambda_0 = (d+2)/(d+4)$, $\lambda_1 = 1-\lambda_0$, and  let $i$ and $j$ be distinct elements in $\{1, 2, \ldots, n\}$.  We have
{\small 
\begin{eqnarray}
\nonumber \tr(\Lambda_{d+2} U_i(d+2)^\dagger U_j(d+2)) 
&=&\lambda_0\left(\frac{1}{(n+1)^2}+\left(1-\frac{1}{(n+1)^2}\right)(\Pm[U_i(d)]_1)^\dagger(\Pm[U_j(d)]_1)\right)+\lambda_1\\
&=&\frac{d+2}{d+4}\left(\frac{1}{(n+1)^2}+\left(\frac{n^2+2n}{(n+1)^2}\right)([U_i(d)]_1^\dagger [U_j(d)]_1\right)+\frac{2}{d+4}.\label{OCE}
\end{eqnarray}
}
Since $U_i(d)$ and $U_j(d)$ are $\Lambda_{d}$-orthogonal at $\lambda_0=\frac{d}{d+2}, \lambda_1=\frac{2}{d+2}$,
$$
\frac{d}{d+2}\left([U_i(d)]_1^\dagger[U_j(d)]_1\right)+\frac{2}{d+2}=0,
$$
from which it follows that $\left.\right.[U_i(d)]_1^\dagger[U_j(d)]_1=-\frac{1}{n}$.  Substituting $-\frac{1}{n}$ for  $\left.\right.[U_i(d)]_1^\dagger[U_j(d)]_1$ in (\ref{OCE}) and doing a bit of algebra reveals that $ \tr(\Lambda_{d+2} U_i(d+2)^\dagger U_j(d+2)) = 0$, as desired.  An argument essentially the same as the preceding one shows that $V_i(d+2)$ and $V_j(d+2)$ are $\Lambda_{d+2}$-orthogonal (given $i$ and $j$ are distinct elements in $\{1, 2, \ldots, n\}$).  Let $j\in \{1, 2, \ldots, n\}$.    It remains to test the $\Lambda_{d+2}$-orthogonality of $U_{n+1}(d+2)$ and $U_j(d+2)$  (as well as of $V_{n+1}(d+2)$  and $V_j(d+2)$).    We have
$$
 \tr(\Lambda_{d+2} U_{n+1}(d+2)^\dagger U_j(d+2)) =  \frac{d+2}{d+4}\left(\frac{1}{(n+1)^2}+\left(1-\frac{1}{(n+1)^2}\right)\left(-\frac{1}{n}\right)\right)+\frac{2}{d+4} = 0
 $$
 and essentially the same calculation shows $ \tr(\Lambda_{d+2} V_{n+1}(d+2)^\dagger V_j(d+2))  = 0$.  This completes the verification that $\F_{(d+2)/(d+4)}$ is a $\Lambda_{d+2}$-orthogonal family of unitary matrices.    It follows that for every even $d \ge 2$,  the bound $\lambda_0 = d/(d+2)$ is saturated.

 We now argue that the bound $\lambda_0 = d/(d+2)$ is saturated when $d \ge 5$ is odd.    We construct $\Lambda_d$-orthogonal families of $d+2$ unitary matrices based on an inductive argument that begins with the $d = 5$ case; i.e., with the family $\F_{5/7}$.  The members of this family will be exhibited in a moment.  
 
First, however, for odd $d \ge 7$, define, inductively, the family $\mathcal{F}_{d/(d+2)}$ of  $d+2$ matrices of size $d\times d$  by 
\begin{equation}\label{ddto}
\mathcal{F}_{d/(d+2)}=\{I_d,A_d,M(d),M(d)^*\}\cup\{U_j(d)\}_{j=1}^{d/2-1/2}\cup\{V_k(d)\}_{k=1}^{d/2-3/2},
\end{equation}
where
\begin{equation}\label{UKVK}
U_{\frac{d}{2}-\frac{1}{2}}(d)=\left[\begin{array}{cc}-\frac{2}{d}&0\\0&1\\\sqrt{1-(\frac{2}{d})^2}\left[I_{\frac{d}{2}-\frac{1}{2}}\right]_1&\bold0_{\frac{d}{2}-\frac{1}{2}} \\ \bold0_{\frac{d}{2}-\frac{3}{2}}&\bold0_{\frac{d}{2}-\frac{3}{2}}\end{array}\right.,
V_{\frac{d}{2}-\frac{3}{2}}(d)=\left[\begin{array}{cc}0&1\\-\frac{2}{d}&0\\\bold0_{\frac{d}{2}-\frac{1}{2}}&\bold0_{\frac{d}{2}-\frac{1}{2}}\\\sqrt{1-(\frac{2}{d})^2}\left[I_{\frac{d}{2}-\frac{3}{2}}\right]_{\frac{d}{2}-\frac{3}{2}}&\bold0_{\frac{d}{2}-\frac{3}{2}}\end{array}\right.,
\end{equation}
and the remaining $d/2-3/2$ matrices in $\{U_j(d)\}_{j=1}^{d/2-1/2}$ and $d/2-5/2$ matrices in $\{V_k(d)\}_{k=1}^{d/2-3/2}$ are constructed from corresponding matrices in the family $\mathcal{F}_{(d-2)/d}$, where
\begin{equation}\label{ULVL}
U_j(d)=\left[\begin{array}{cc}-\frac{2}{d}&0\\0&1\\\sqrt{1-(\frac{2}{d})^2}\Pm[U_j(d-2)]_1&\bold0_{d-2}\end{array}\right.,
V_k(d)=\left[\begin{array}{cc}0&1\\-\frac{2}{d}&0\\\sqrt{1-(\frac{2}{d})^2}\Pm[V_k(d-2)]_1&\bold0_{d-2}\end{array}\right. .
\end{equation}
For odd $d$, the family $\mathcal{F}_{d/(d+2)}$ will always contain two ``$M$'' matrices,  $M(d)$ and $M(d)^*$, where
\begin{equation}\label{M}
M(d)=\left[\begin{array}{cc}-\frac{1}{d}&\frac{\sqrt{3}}{2}i\\-\frac{\sqrt{3}}{d}i&\frac{1}{2}  \\ \sqrt{1-(\frac{2}{d})^2}\Pm[M(d-2)]_1&\bold0_{d-2}\end{array}\right.\ .
\end{equation}

 Our inductive proof begins with $\mathcal{F}_{5/7}$, consisting of $I_5,A_5,$
$$
U_1(5)=\left[\begin{array}{cc}-\frac{2}{5}&0\\0&1\\\frac{\sqrt{21}}{5}\left[\begin{array}{c}-\frac{2}{3}\\\frac{\sqrt{5}}{3}\end{array}\right]&\begin{array}{c}0\\0\end{array}\\0&0\end{array}\right. ,
U_2(5)=\left[\begin{array}{cc}-\frac{2}{5}&0\\0&1\\\frac{\sqrt{21}}{5}&0\\0&0\\0&0\end{array}\right. ,
V_1(5)=\left[\begin{array}{cc}0&1\\-\frac{2}{5}&0\\0&0\\0&0\\\frac{\sqrt{21}}{5}&0\end{array}\right. ,
$$
$$
M(5)=\left[\begin{array}{cc}-\frac{1}{5}&\frac{\sqrt{3}}{2}i\\-\frac{\sqrt{3}}{5}i&\frac{1}{2}\\\frac{\sqrt{21}}{5}\left[\begin{array}{c}-\frac{1}{3}\\-\frac{\sqrt{5}}{3}\\-\frac{\sqrt{3}}{3}i\end{array}\right]&\begin{array}{c}0\\0\\0\end{array}\end{array}\right.,
M(5)^*=\left[\begin{array}{cc}-\frac{1}{5}&-\frac{\sqrt{3}}{2}i\\\frac{\sqrt{3}}{5}i&\frac{1}{2}\\\frac{\sqrt{21}}{5}\left[\begin{array}{c}-\frac{1}{3}\\-\frac{\sqrt{5}}{3}\\\frac{\sqrt{3}}{3}i\end{array}\right]&\begin{array}{c}0\\0\\0\end{array}\end{array}\right. .
$$
The reader may verify that this family is $\Lambda_5$-orthogonal and consists of unitary matrices. Suppose that for some odd $d\ge5$ the family $\mathcal{F}_{d/(d+2)}=\{I_d,A_d,M(d),M(d)^*\}\cup\{U_j(d)\}_{j=1}^{d/2-1/2}\cup\{V_k(d)\}_{k=1}^{d/2-3/2}$ is $\Lambda_d$-orthogonal and consists of unitaries. We claim that the family $\mathcal{F}_{(d+2)/(d+4)}=\{I_{d+2},A_{d+2},M(d+2),M(d+2)^*\}\cup\{U_j(d+2)\}_{j=1}^{d/2+1/2}\cup\{V_k(d+2)\}_{k=1}^{d/2-1/2}$ is a $\Lambda_{d+2}$-orthogonal family of unitaries, and establishing this claim completes the proof.

The first two columns of $U_{\frac{d}{2}+\frac{1}{2}}(d+2)$ and $V_{\frac{d}{2}-\frac{1}{2}}(d+2)$, defined by (\ref{UKVK}), are easily seen to be orthonormal, and thus can be augmented to create unitary matrices. For $U_j(d+2)$ and $V_k(d+2)$, defined by (\ref{ULVL}), and $M(d+2)$ and $M(d+2)^*$, defined by (\ref{M}), orthogonality of the first two columns can be easily verified, and normalization follows from the fact that $\Pm[U_j(d)]_1, \Pm[V_k(d)]_1$, and $ \Pm[M(d)]_1$ are unit vectors. Now we must verify $\Lambda_{d+2}$-orthogonality among the members of $\mathcal{F}_{(d+2)/(d+4)}$.

Just as in the case for even $d$, the matrices $I_{d+2}$ and $A_{d+2}$ are $\Lambda_{d+2}$-orthogonal, and once again, due to the placement of zeros, each $U_j(d+2)$ matrix will be $\Lambda_{d+2}$-orthogonal to each $V_k(d+2)$ matrix, for $j\in \{1, 2, \ldots, d/2 +1/2\}, k\in \{1,2,\ldots, d/2-1/2\}$ .  It is easy to check that $M(d+2), M(d+2)^*$ as well as each $U_j(d+2)$ and $V_k(d+2)$ matrix is $\Lambda_{d+2}$-orthogonal to both $I_{d+2}$ and  $A_{d+2}$.

Our next step is the verification that $\{U_j(d+2)\}_{j=1}^{d/2+1/2}$ is a  $\Lambda_{d+2}$-orthogonal family.   Let $i$ and $j$ be distinct elements in $\{1,2,\ldots,d/2-1/2\}$; we have
\begin{eqnarray}
\nonumber \tr(\Lambda_{d+2}U_i(d+2)^\dag U_j(d+2))
&=&\lambda_0\left(\frac{4}{(d+2)^2}+\left(1-\left(\frac{2}{d+2}\right)^2\right)(\Pm[U_i(d)]_1)^\dag\Pm[U_j(d)]_1\right)+\lambda_1\\
&=&\frac{d+2}{d+4}\left(\frac{4}{(d+2)^2}+\frac{d(d+4)}{(d+2)^2}[U_i(d)]_1^\dag[U_j(d)]_1\right)+\frac{2}{d+4}.\label{UIJ}
\end{eqnarray}
Since $U_i(d)$ and $U_j(d)$ are $\Lambda_{d}$ orthogonal at $\lambda_0=\frac{d}{d+2}, \lambda_1=\frac{2}{d+2}$,
$$
\frac{d}{d+2}\left([U_i(d)]_1^\dagger[U_j(d)]_1\right)+\frac{2}{d+2}=0,
$$
and it follows that $[U_i(d)]_1^\dagger[U_j(d)]_1=-\frac{2}{d}$. Substituting $-\frac{2}{d}$ for $[U_i(d)]_1^\dagger[U_j(d)]_1$ into  (\ref{UIJ}) and simplifying the resulting equation will show that $\tr(\Lambda_{d+2}U_i(d+2)^\dag U_j(d+2))=0,$ and so for $j\in\{1,2,\ldots,d/2-1/2\}$, each $U_j(d+2)$ matrix is $\Lambda_{d+2}$-orthogonal to every other $U_j(d+2)$. A similar approach may be used to prove $V_i(d+2)$ and $V_j(d+2)$ are $\Lambda_{d+2}$-orthogonal for two distinct elements $i,j\in\{1,2,\ldots,d/2-3/2\}$.  As for $\Lambda_{d+2}$-orthogonality of $U_{\frac{d}{2}+\frac{1}{2}}(d+2)$  and  $U_j(d+2)$, for  $j\in\{1,2,\ldots,d/2-1/2\}$:
$$
\tr(\Lambda_{d+2}U_{\frac{d}{2}+\frac{1}{2}}(d+2)^\dag U_j(d+2))= \frac{d+2}{d+4}\left(\frac{4}{(d+2)^2}+\left(1-\left(\frac{2}{d+2}\right)^2\right)\left(-\frac{2}{d}\right)\right)+\frac{2}{d+4},
$$
which simplifies to 0.  A similar calculation shows that $\tr(\Lambda_{d+2}V_{\frac{d}{2}-\frac{1}{2}}(d+2)^\dag V_j(d+2))=0$ for each $j\in\{1,2,\ldots,d/2-3/2\}$. Thus the set $\{I_{d+2},A_{d+2}\}\cup\{U_j(d+2)\}_{j=1}^{d/2+1/2}\cup\{V_j(d+2)\}_{j=1}^{d/2-1/2}$ is $\Lambda_{d+2}$-orthogonal, and all that remains to be shown is that this set, augmented with $\{M(d+2),M(d+2)^*\}$, remains $\Lambda_{d+2}$-orthogonal.

We begin by showing that  $M(d+2)^*$ and $M(d+2)$ are  $\Lambda_{d+2}$-orthogonal. We have
\begin{equation}
\begin{split}
\tr(\Lambda_{d+2}(M(d+2)^*)^\dagger M(d+2)) \rule{3.5in}{0in} \\=\frac{d+2}{d+4}\left(-\frac{2}{(d+2)^2}+\frac{d(d+4)}{(d+2)^2}([M(d)^*]_1^\dag[M(d)]_1)\right)+\frac{2}{d+4}\left(-\frac{1}{2}\right)\label{M1M2}.
\end{split}
\end{equation}
Since $M(d)$ and $M(d)^*$ are $\Lambda_d$-orthogonal when $\lambda_0=\frac{d}{d+2},\lambda_1=\frac{2}{d+2},$
$$
\frac{d}{d+2}([M(d)^*]_1^\dag[M(d)]_1)+\frac{2}{d+2}\left(-\frac{1}{2}\right)=0
$$
so $[M(d)^*]_1^\dag[M_1(d)]_1=\frac{1}{d}$. If we substitute $\frac{1}{d}$ for $[M(d)^*]_1^\dag[M(d)]_1$ in (\ref{M1M2})  and simplify we see $\tr(\Lambda_{d+2}(M(d+2)^*)^\dag M(d+2))=0$, as desired.

 We now establish the $\Lambda_{d+2}$-orthogonality of  the $U_j(d+2)$ matrices with  $M(d+2)$ and $M(d+2)^*$.    Let $j\in\{1,2,\ldots,d/2-1/2\}$; we have
 \begin{equation}\label{UM}
\begin{split}
\tr(\Lambda_{d+2}U_j(d+2)^\dag M(d+2)) \rule{3.5in}{0in}  \\ = \frac{d+2}{d+4}\left(\frac{2}{(d+2)^2}+\left(\frac{d(d+4)}{(d+2)^2}\right)[U_j(d)]_1^\dag[M(d)]_1\right)+\frac{2}{d+4}\left(\frac{1}{2}\right).
\end{split}
\end{equation}
Because $U_j(d)$ and $M(d)$ are $\Lambda_d$-orthogonal at $\lambda_0=\frac{d}{d+2} ,\lambda_1=\frac{2}{d+2}$, 
$$
\frac{d}{d+2}([U_j(d)]_1^\dag[M(d)]_1)+\frac{2}{d+2}\left(\frac{1}{2}\right)=0.
$$
Hence $[U_j(d)]_1^\dag[M(d)]_1=-\frac{1}{d}$, and by making this substitution in (\ref{UM}) and simplifying, $U_j(d+2)$ is seen to be 
$\Lambda_{d+2}$-orthogonal to $M(d+2)$. Also, 
$$
\tr(\Lambda_{d+2}U_{\frac{d}{2}+\frac{1}{2}}(d+2)^\dag M(d+2))=\frac{d+2}{d+4}\left(\frac{2}{(d+2)^2}+\left(1-\left(\frac{2}{d+2}\right)^2\right)\left(-\frac{1}{d}\right)\right)+\frac{2}{d+4}\left(\frac{1}{2}\right),
$$
which simplifies to 0; so we have shown that every $U_j(d+2)$ matrix is $\Lambda_{d+2}$-orthogonal to $M(d+2)$. Since $U_j(d+2)$ and $\Lambda_{d+2}$ are real, $\tr(\Lambda_{d+2}U_j(d+2)^\dag M(d+2)) = 0$ immediately implies $\tr(\Lambda_{d+2}U_j(d+2)^\dag M(d+2)^*) = 0$; thus, every $U_j(d+2)$ matrix is also orthogonal to  $M(d+2)^*$.

We complete the argument by establishing the $\Lambda_{d+2}$-orthogonality of the $V_j(d+2)$ matrices to $M(d+2)$ and $M(d+2)^*$.  
Let  $j\in\{1,2,\ldots,d/2-3/2\}$; we have
\begin{equation}\label{VM}
\begin{split}
\tr(\Lambda_{d+2}V_j(d+2)^\dag M(d+2)) \rule{3.5in}{0in}  \\= \frac{d+2}{d+4}\left(\frac{2\sqrt{3}i}{(d+2)^2}+\left(\frac{d(d+4)}{(d+2)^2}\right)[V_j(d)]_1^\dag[M(d)]_1\right)+\frac{2}{d+4}\left(\frac{\sqrt{3}}{2}i\right).
\end{split}
\end{equation}
We are assuming that  $V_j(d)$ and $M(d)$ are $\Lambda_d$-orthogonal at $\lambda_0=d/(d+2),\lambda_1=1-\lambda_0$, so
\begin{equation}\label{VJM1}
\frac{d}{d+2}([V_j(d)]_1^\dag[M(d)]_1)+\frac{2}{d+2}\left(\frac{\sqrt{3}}{2}i\right)=0.
\end{equation}
Then $[V_j(d)]_1^\dag[M(d)]_1=-\frac{\sqrt{3}i}{d}$, and by making this substitution in (\ref{VM}) and simplifying, $V_j(d+2)$ is seen to be 
$\Lambda_{d+2}$-orthogonal to $M(d+2)$. Also,
\begin{equation}\label{VJM1_2}
\begin{split}
\tr(\Lambda_{d+2}V_{\frac{d}{2}-\frac{1}{2}}(d+2)^\dag M(d+2))\rule{3in}{0in} \\=\frac{d+2}{d+4}\left(\frac{2\sqrt{3}i}{(d+2)^2}+\left(1-\left(\frac{2}{d+2}\right)^2\right)\left(-\frac{\sqrt{3}i}{d}\right)\right)+\frac{2}{d+4}\left(\frac{\sqrt{3}}{2}i\right),
\end{split}
\end{equation}
which simplifies to 0, so we have shown that each $V_j(d+2)$ matrix is $\Lambda_{d+2}$-orthogonal to $M(d+2)$.   Because $\Lambda_{d+2}$ and $V_j(d+2)$ are real, it  follows that every $V_j(d+2)$ matrix is also $\Lambda_{d+2}$-orthogonal to $M(d+2)^*$.   

We have shown that  $M(d+2)$ and $M(d+2)^*$  are orthogonal to each other and to each matrix in the set $\{I_{d+2},A_{d+2}\}\cup\{U_j(d+2)\}_{j=1}^{d/2+1/2}\cup\{V_k(d+2)\}_{k=1}^{d/2-1/2}$.  Therefore, the family $\mathcal{F}_{(d+2)/(d+4)}=\{I_{d+2},A_{d+2},M(d+2),M(d+2)^*\}\cup\{U_j(d+2)\}_{j=1}^{d/2+1/2}\cup\{V_k(d+2)\}_{k=1}^{d/2-1/2}$ is composed of d+4 $\Lambda_{d+2}$-orthogonal matrices. This completes the argument that WCSG bound $\lambda_0 = d/(d+2)$ is saturated for every odd $d$.


\end{document}